\documentclass[aps,prx,twocolumn,showpacs,amsmath,amssymb,floatfix,superscriptaddress]{revtex4-2}

\usepackage[outline]{contour}

\usepackage{color}
\usepackage{bbm}
\usepackage{graphicx} 

\usepackage{dcolumn}
\usepackage[utf8]{inputenc}
\usepackage{graphicx}
\usepackage{epstopdf}
\usepackage{amsthm}

\usepackage{amsmath}
\usepackage{empheq}
\usepackage{float}

\usepackage{eufrak}
\usepackage{bbm}
\usepackage{braket}
\usepackage{amssymb}
\usepackage{mathtools}
\usepackage[usenames,dvipsnames]{xcolor}
\usepackage{tikz}
\usepgflibrary{shapes.arrows}
\usetikzlibrary{calc, positioning, shapes.arrows}

\usepackage{cases}
\usepackage{smartdiagram}
\usepackage{latexsym}
\usepackage[colorlinks=true,citecolor=Cerulean,linkcolor=RubineRed,urlcolor=Cerulean]{hyperref}

\usepackage{cleveref}

\contourlength{0.5pt}
\contournumber{20}

\definecolor{S_Blue}{RGB}{0,135,252}
\definecolor{S_Red}{RGB}{214,13,63}
\definecolor{Blue}{RGB}{47,89,151}
\definecolor{S_Grey}{RGB}{150,150,158}
\definecolor{S_Yel}{RGB}{255,204,0}

\definecolor{S_Green}{RGB}{102,204,0}
\definecolor{S_Brown}{RGB}{154,41,41}

\DeclareFontFamily{OT1}{pzc}{}
\DeclareFontShape{OT1}{pzc}{m}{it}{<-> s * [1.15] pzcmi7t}{}
\DeclareMathAlphabet{\mathpzc}{OT1}{pzc}{m}{it}

\graphicspath{ {figures/figure1/}{figures/figure2/} }

\DeclareMathOperator{\sech}{sech}

\newtheorem{lemma}{Lemma}

\newcommand{\M}{\textnormal{max}}
\newcommand{\m}{\textnormal{min}}

\newcommand{\info}{\mathcal{F}} 

\newcommand{\id}{\mathbbm{1}} 
\newcommand{\tr}[1]{\operatorname{\textnormal{Tr}}\left( {#1} \right)} 
\newcommand{\cov}{\textnormal{cov}}
\newcommand{\var}{\textnormal{var}}

\newcommand{\gM}{\|H\|_{\textnormal{s}}}
\newcommand{\gm}{g_\textnormal{min}}

\definecolor{KB}{rgb}{0.4,0.3,0.9}

\begin{document}

\title{Estimation of Hamiltonian parameters from thermal states}

\author{Luis Pedro Garc\'ia-Pintos}
\email{lpgp@lanl.gov}
\affiliation{Theoretical Division (T4), Los Alamos National Laboratory, Los Alamos, New Mexico 87545, USA}

\author{Kishor Bharti}
\affiliation{Joint Center for Quantum Information and Computer Science and Joint Quantum Institute, NIST/University of Maryland, College Park, Maryland 20742, USA}
\affiliation{A*STAR Quantum Innovation Centre (Q.InC), Institute of High Performance Computing (IHPC), Agency for Science,
Technology and Research (A*STAR), 1 Fusionopolis Way,
\# 16-16 Connexis, Singapore 138632, Republic of Singapore}

\author{Jacob Bringewatt}
\affiliation{Joint Center for Quantum Information and Computer Science and Joint Quantum Institute, NIST/University of Maryland, College Park, Maryland 20742, USA}

\author{Hossein Dehghani}
\affiliation{Joint Center for Quantum Information and Computer Science and Joint Quantum Institute, NIST/University of Maryland, College Park, Maryland 20742, USA}
\affiliation{Department of Physics, University of Maryland, College Park, Maryland 20742, USA}

\author{Adam Ehrenberg}
\affiliation{Joint Center for Quantum Information and Computer Science and Joint Quantum Institute, NIST/University of Maryland, College Park, Maryland 20742, USA}

\author{Nicole Yunger Halpern}
\affiliation{Joint Center for Quantum Information and Computer Science and Joint Quantum Institute, NIST/University of Maryland, College Park, Maryland 20742, USA}
\affiliation{
Institute for Physical Science and Technology, University of Maryland, College Park, MD 20742, USA}

\author{Alexey V. Gorshkov}
\affiliation{Joint Center for Quantum Information and Computer Science and Joint Quantum Institute, NIST/University of Maryland, College Park, Maryland 20742, USA}


\begin{abstract}

We upper- and lower-bound the optimal precision with which one can estimate an unknown Hamiltonian parameter via measurements of Gibbs thermal states with a known temperature. The bounds depend on the uncertainty in the Hamiltonian term that contains the parameter and on the term's degree of noncommutativity with the full Hamiltonian: 
higher uncertainty and commuting operators lead to better precision. We apply the bounds to show that there exist entangled thermal states such that the parameter can be estimated with an error that decreases faster than $1/\sqrt{n}$, beating the standard quantum limit. This result governs Hamiltonians where an unknown scalar parameter (e.g. a component of a magnetic field) is coupled locally and identically to $n$ qubit sensors.
In the high-temperature regime, our bounds allow for pinpointing the optimal estimation error, up to a constant prefactor.
Our bounds generalize to joint estimations of multiple parameters. In this setting, we recover the high-temperature sample scaling derived previously via techniques based on quantum state discrimination and coding theory.
In an application, we show that noncommuting conserved quantities hinder the estimation of chemical potentials. 

\end{abstract}

\maketitle

Substantial work has been devoted to determining the precision with which a Hamiltonian parameter can be estimated from measurements on a time-evolving system. For instance, consider a spin network immersed in a magnetic field $\mu$. The network's state acquires information about the field's magnitude. Measuring copies of the state can reveal $\mu$. The quantum Cram\'er-Rao bound sets an asymptotically saturable lower bound on the precision with which the parameter can be estimated~\cite{helstrom1969quantum,
 BraunsteinCaves1994}.

Here, we focus on the less explored problem of estimating parameters from systems in a thermal state 
\begin{align}
\label{eq:Gibbs}
\rho = \frac{1}{Z_\beta} e^{-\beta H} = \frac{1}{Z_\beta} \sum_j e^{-\beta \omega_j} \ket{j}\!\bra{j}
\end{align}
at a known inverse temperature $\beta$.
 $\omega_j$ and $\ket{j}$ are the Hamiltonian's eigenvalues and eigenvectors, $H \nobreak=\nobreak \sum_j \omega_j \ket{j}\!\bra{j}$, and $Z_\beta \coloneqq \tr{e^{-\beta H}}$ is the partition function. Parameters of $H$ could be unknown. The system could thermalize to $\rho$ through interactions with a thermal environment or through a state-preparation algorithm~\cite{Yigit,kastoryano2023quantum}. Probing the environment could yield information about $\beta$~\cite{CorreaThermoPRL2015}. 

The thermal state encodes information about the Hamiltonian parameters. We consider $M$-term Hamiltonians:
\begin{align}
\label{eq:Hamiltonian}
H = \sum_{l=1}^M H_l = \sum_{l=1}^M \mu_l A_l.
\end{align}
The $A_l$ are Hermitian operators, and the $\mu_l$ are real coefficients. 
The $\mu_l$ could represent local or global fields or coupling constants (Fig.~\ref{fig:fig1}).
We bound the precision with which the $\mu_l$ can be estimated from measurements of copies of 
$\rho$. 
To achieve this goal, we will use the multiparameter quantum Cram\'er-Rao bound, which constrains the estimation of a set of parameters~\cite{liu2019quantum}.

The quantum Cram\'er-Rao bound
relates the minimum estimation error to the quantum Fisher information~\cite{paris2009quantum}. The quantum Cram\'er-Rao bound has been applied, for example, to the field of thermometry
~\cite{Paris_2015,
CorreaThermoPRL2015,
CorreaThermoPRB2018,
miller_energytemperature,
Potts2019fundamentallimits,
mok2021optimal,PhysRevLett.128.130502}. The bound implies the minimum uncertainty with which a temperature $T$ can be estimated from $\mathcal{N}$ measurements: $\var(\hat{T}_\textnormal{opt}) \nobreak = \nobreak T^4 \tfrac{1}{\mathcal{N} (\Delta H)^2}$, where $(\Delta H)^2 \nobreak \coloneqq \nobreak \langle H^2 \rangle \nobreak-\nobreak \langle H \rangle^2$ is the Hamiltonian's variance in the thermal state~\cite{liu2019quantum}. Whenever $x$ denotes a parameter to be estimated, we mean by $\hat{x}$ an estimator. Higher energy variances allow for better parameter estimation.
This result echoes the relative error $\var(\hat{\mu}_\textnormal{opt}) /\mu^2\nobreak = \nobreak \tfrac{1}{4 \mathcal{N} t^2 (\Delta H)^2}$ with which a global parameter $\mu$ can be estimated from measurements of copies of a pure state evolving under the Hamiltonian $H = \mu A$ for a time $t$.
In related work, Refs.~\cite{ZanardiPRA2007,ZanardiPRA2017b} geometrically characterize the Fisher metric to study the role of phase transitions in thermometry. 
This Letter focuses on the error in estimates of an \emph{arbitrary} Hamiltonian parameter, rather than the error in temperature estimation.

Several studies have concerned the
reconstruction of a Hamiltonian from its eigenstates~\cite{garrison2018does,
 greiter2018method,
 chertkov2018computational,
 zhu2019reconstructing,
 qi2019determining,
 turkeshi2019entanglement,
 bairey2019learning}, from steady states~\cite{zhou2021can}, or from Gibbs states~\cite{SensingHamPRA2015,bairey2019learning}. 
Recent results under the umbrella of the ``Hamiltonian-learning problem'' provide algorithms for estimating Hamiltonian parameters while minimizing (i) the number of copies of the thermal state $\rho$ needed (the sample complexity) and (ii) the algorithm's runtime (the time complexity)~\cite{
anshu2021sample,
haah2022optimal,
gu2022practical,
sbahi2022provably}.
Such complexity-theoretic approaches focus on (a) the asymptotic sample and time complexities' dependence on $\beta$ and (b) the number of unknown parameters. In contrast, we leverage the quantum Cram\'er-Rao bound to identify how 
the uncertainties in the $A_l$s, and the $A_l$s' noncommutativity with the thermal state, influence the minimum precision with which the $\mu_l$ can be estimated. 
Upon pinpointing the uncertainties' influence on precision, we can construct a many-body model that beats the standard quantum limit.

This Letter is organized as follows. First, we review the quantum Fisher information, a powerful tool for analyzing parameter estimation. We bound the quantum Fisher information obtainable about one Hamiltonian parameter, then bound the precision with which the parameter can be estimated. These bounds enable us to identify a many-body model in which the achievable precision beats the standard quantum limit. Extending beyond one Hamiltonian parameter, we then bound the precision with which multiple parameters can be estimated simultaneously. Finally, we discover that noncommutation of conserved quantities (\emph{charges}) hinders the estimation of chemical potentials. Noncommuting charges are particularly quantum (due to the importance of noncommutation in quantum measurement disturbance, Heisenberg uncertainty, etc.) and have been of recent thermodynamic interest~\cite{Majidy_23_Noncommuting}.

\begin{figure}
  \centering  
   \includegraphics[trim=00 00 00 00,width=0.45 \textwidth]{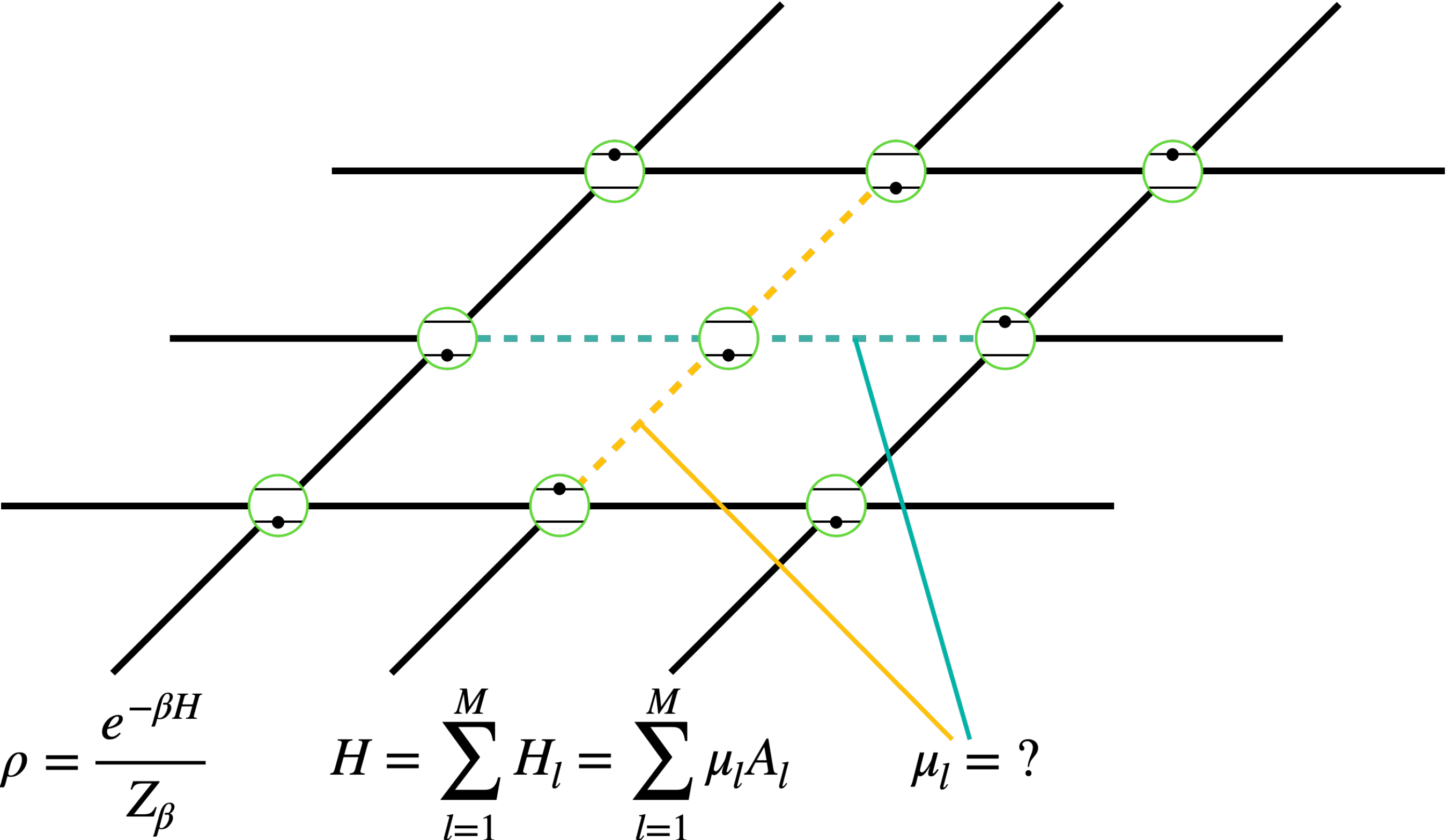}
\caption{\textbf{Estimating Hamiltonian parameters from thermal states.} 
How accurately can one determine $\mu_l$, which can be a coupling constant---pictured here as yellow and teal dashed lines for a system of qubits on a lattice---or a field, in a Hamiltonian $H$ from measurements performed on $\mathcal{N}$ copies of a thermal state $\rho = e^{-\beta H}/Z_\beta$? 
We use the quantum Cram\'er-Rao bound to derive saturable upper and lower bounds on the optimal precision with which such Hamiltonian parameters can be estimated.
\label{fig:fig1}
}
\end{figure}

\vspace{7pt}
\textit{The quantum Fisher information matrix.---}The multiparameter quantum Cram\'er-Rao bound constrains the statistics of any estimator $\hat{\vec\mu}$ of the parameters $\mu_l$~\cite{liu2019quantum}:
\begin{align}
    \label{eq:multiCR}
    \cov(\hat{\vec{\mu}}) \geq \frac{1}{\mathcal{N}}\info^{-1}.
\end{align}
$\mathcal{N}$ denotes the number of experimental repetitions.  
$\info$ denotes the quantum Fisher information matrix, with components
\begin{align}
\label{eq:FisherMatrix}
    \info_{lm} \coloneqq 2 \sum_{jk} \frac{ \textnormal{Re} \left[ \bra{j} \partial_l \rho \ket{k}\!  \bra{k} \partial_m \rho \ket{j} \right] }{p_j + p_k}.
\end{align}
The state eigendecomposes as $\rho = \sum_j p_j \ket{j}\!\bra{j}$.
Thus, the quantum Fisher information matrix characterizes the precision with which parameters $\mu_l$ can be estimated jointly.  
The multiparameter quantum Cram\'er-Rao bound is saturated when the optimal measurements for estimating the $\mu_l$ are compatible. Mathematically, this condition is met if and only if 
$    \tr{\rho [L_l,L_m]}
\nobreak = \nobreak 0$.
The symmetric logarithmic derivative $L_l$ is implicitly defined by $\partial_l \rho \eqqcolon \tfrac{1}{2}\{ \rho,L_l \}$~\cite{liu2019quantum}.
Throughout this work, we denote partial derivatives by $\partial_l \coloneqq \tfrac{\partial}{\partial \mu_l}$.

The diagonal matrix element $\info_{ll}$ quantifies the minimum precision with which one unknown $\mu_l$ can be estimated if all other parameters are known.
The single-parameter quantum Cram\'er-Rao bound says that every estimator $\hat{\mu_l}$ has a variance 
\begin{align}
\label{eq:CRsingle}
\var(\hat{\mu_l}) \geq \frac{1}{\mathcal{N} \info_{ll}}.
\end{align}
Optimized measurements saturate this bound~\cite{BraunsteinCaves1994,
paris2009quantum}. Equations~\eqref{eq:multiCR} and~\eqref{eq:CRsingle} thus pinpoint the quantum Fisher information as a powerful tool that determines ultimate limits on quantum metrology. The stronger $\rho$'s dependency on $\mu_l$, the higher the quantum Fisher information $\info_{ll}$ [Eq.~\eqref{eq:FisherMatrix}], and so the greater the precision.

\vspace{7pt}
\textit{Bounds on the quantum Fisher information.---}Exactly evaluating the quantum Fisher information can be difficult. Therefore, it is desirable to bound $\info_{ll}$ in terms of more-easily-calculable quantities. We derive two sets of upper and lower bounds on the quantum Fisher information of the $\mu_l$ in Eq.~\eqref{eq:Hamiltonian}:
\begin{subequations}
 \label{eq:BoundsDiagonalDeltaH}
 \begin{align}
\label{eq:UpperBoundsDiagonalDeltaH}
  \info_{ll}  &\leq \beta^2 \, (\Delta A_l)^2  \\ 
\label{eq:LowerBoundsDiagonalDeltaH}
  \info_{ll} &\geq 4  \beta^2 c_1  \, (\Delta A_l)^2 ,  
\end{align}   
\end{subequations}
and   
\begin{subequations}
\label{eq:BoundsDiagonalWY}
\begin{align}
\label{eq:BoundsDiagonalWYupper}
\info_{ll} &\leq 2.4 \, c_2 \, \beta^2 \, \left( (\Delta A_l)^2 \nobreak-\nobreak \tfrac{1}{2}\big\|  [\sqrt{\rho},A_l] \big\|_2^2 \right)
\quad \text{and} \\
\label{eq:BoundsDiagonalWYlower}
\info_{ll} &\geq 0.8 \, \beta^2 \left( (\Delta A_l)^2 \nobreak-\nobreak \tfrac{1}{2}\big\|  [\sqrt{\rho},A_l] \big\|_2^2 \right).
\end{align}
\end{subequations}
$\Delta A_l \nobreak = \nobreak \sqrt{\langle A_l^2 \rangle - \langle A_l \rangle^2}$ is the uncertainty of operator $A_l$ in $\rho$; $\| A \|_2^2 \coloneqq \tr{A A^\dag}$; and we have defined 
\begin{subequations}
\label{eq:cConst}
    \begin{align}
c_1 &\coloneqq  \tanh^2(\beta \gM/2)/(\beta \gM)^2
\quad \text{and}
\\
c_2 &\coloneqq 2 c_1 \, \cosh(\beta \gM/2) . 
\end{align} 
\end{subequations}
The $\gM \nobreak \coloneqq \nobreak  \max_j \omega_j - \min_j \omega_j$ is the Hamiltonian seminorm defined by the maximum energy gap.
We derive the bounds by computing the thermal state's quantum Fisher information, then algebraically manipulating the expression (Appendix~\ref{sec-app:BoundsOnFisher}).

Equation~\eqref{eq:BoundsDiagonalDeltaH} constrains the quantum Fisher information about $\mu_l$ in terms of $\Delta A_l$, resembling expressions for the quantum Fisher information about $\beta$ in thermometry~\cite{liu2019quantum}. Equation~\eqref{eq:BoundsDiagonalWY} constrains the quantum Fisher information about $\mu_l$ also in terms of the Wigner-Yanase skew information $\tfrac{1}{2}\big\| [\sqrt{\rho},A_l] \big\|_2^2$. The skew information was proposed as a means to discriminate quantum and classical contributions to 
uncertainty~\cite{wigner1963information,luo2005quantum}. It has found applications in parameter estimation~\cite{LuoPRL2003, miller_energytemperature,sidhu2020geometric}, as an asymmetry measure~\cite{marvian2014extending}, and as a coherence measure~\cite{PhysRevLett.113.170401,PhysRevA.91.042330}. The difference $(\Delta A_l)^2 \nobreak-\nobreak \tfrac{1}{2}\big\|  [\sqrt{\rho},A_l] \big\|_2^2$ signifies the classical uncertainty about $A_l$~\cite{luo2005quantum}. 
This classical uncertainty vanishes for pure states.

When the temperature is high relative to the maximum energy gap ($\beta \gM \ll 1$), $c_1 \approx c_2/2 \approx 1/4 $. 
The upper and lower bounds in Eq.~\eqref{eq:BoundsDiagonalDeltaH} coincide, while the upper and lower bounds in Eq.~\eqref{eq:BoundsDiagonalWY} differ by a prefactor of $1.2$. That is, our bounds are saturated, up to a constant prefactor, at high temperatures. Our bounds pinpoint $\info_{ll}$ by tightly sandwiching it.

The upper bound~\eqref{eq:BoundsDiagonalWYupper} is also saturable, up to a constant prefactor, at low temperatures. To show this, we denote by $\mu$ the magnitude of a field $\mu \sigma_z$ acting on a qubit with a Hamiltonian
$H \nobreak=\nobreak \Omega_x \sigma_x \nobreak+\nobreak \Omega_z \sigma_z \nobreak+\nobreak \mu \sigma_z$. The $\sigma_\alpha$s are Pauli matrices.
The quantum Fisher information and its upper bound~\eqref{eq:BoundsDiagonalWYupper} can be calculated exactly. 
At
low temperatures ($\beta \|H\|_s \gg 1$), 
$\info_\mu \nobreak\approx\nobreak 16 \Omega_x^2/\|H\|_s^4 \nobreak \leq \nobreak 2.4\, c_2 \beta^2 \left( (\Delta A_l)^2 \nobreak-\nobreak \tfrac{1}{2}\big\|  [\sqrt{\rho},A_l] \big\|_2^2 \right)  \approx  19.2 \Omega_x^2/\|H\|_s^4$; and,
at
high temperatures ($\beta \|H\|_s\ll 1$), $\info_\mu \nobreak\approx\nobreak \beta^2 \nobreak \leq \nobreak 2.4 c_2 \beta^2 \left( (\Delta A_l)^2 \nobreak-\nobreak \tfrac{1}{2}\big\|  [\sqrt{\rho},A_l] \big\|_2^2 \right) \nobreak\approx\nobreak 1.2 \beta^2$
(Appendix~\ref{sec-app:QubitExample}).
The contribution of the Wigner-Yanase skew information is necessary for obtaining a saturable bound at low temperatures.

Reference~\cite{miller_energytemperature} contains the closest previous result:  $\info_{ll} \nobreak \leq \nobreak \beta^2 \int_0^1 \tr{\rho^a \delta A_l \rho^{1-a} \delta A_l} da$, with $\delta A_l \coloneqq A_l \nobreak-\nobreak \langle A_l \rangle$.
Yet our upper bounds~\eqref{eq:UpperBoundsDiagonalDeltaH},~\eqref{eq:BoundsDiagonalWYupper}, and the bound in Ref.~\cite{miller_energytemperature} are different: no bound is tighter than another in all regimes.  To our knowledge, Eqs.~\eqref{eq:LowerBoundsDiagonalDeltaH} and~\eqref{eq:BoundsDiagonalWYlower} are the first lower bounds on thermal states' quantum Fisher information. We compare the bounds in a spin-chain example in Appendix~\ref{sec-app:Comparisons}.

\vspace{7pt}
\textit{Bounds on single-parameter estimation errors.---}Consider estimating an unknown parameter $\mu_l$.
We denote the  
\emph{optimal} error by $\sqrt{\var_{\textnormal{opt}}(\hat{\mu_l})}$. The single-parameter quantum Cram\'er-Rao bound~\eqref{eq:CRsingle} is saturable by suitably chosen estimators~\cite{paris2009quantum}. Therefore, Eqs.~\eqref{eq:BoundsDiagonalDeltaH} and~\eqref{eq:BoundsDiagonalWY} engender two sets of upper and lower bounds on $\sqrt{\var_{\textnormal{opt}}(\hat{\mu_l})}$. The relative error $\frac{ \sqrt{\var_{\textnormal{opt}}(\hat{\mu_l})} }{|\mu_l|}$ 
achievable with $\mathcal{N}$ copies of a thermal state is 
\begin{align}
\label{eq:BoundSingleVariance}
\frac{1}{ \beta \sqrt{\mathcal{N}}
 \, \Delta H_l} 
 \leq \frac{ \sqrt{\var_{\textnormal{opt}}(\hat{\mu_l})} }{|\mu_l|} 
 \leq \frac{1}{2 \beta c_1^{1/2} \sqrt{\mathcal{N}} \, \Delta H_l},
\end{align}
             and
\begin{align}
\label{eq:singleparamBound}
&\frac{1}{\sqrt{2.4c_2} \beta \, \sqrt{\mathcal{N}} \bigg( \left( \Delta H_l \right)^2 -   \tfrac{1}{2} \Big\| \big[\sqrt{\rho}, H_l \big] \Big\|_2^2 \bigg)^{1/2} } \nonumber \\
&\qquad \qquad \leq \frac{ \sqrt{\textnormal{var}_{\textnormal{opt}} ( \hat\mu_l ) }}{|\mu_l|} \leq \\ 
&\frac{1}{\sqrt{0.8} \beta \, \sqrt{\mathcal{N}} \bigg( \left( \Delta H_l \right)^2 -   \tfrac{1}{2} \Big\| \big[\sqrt{\rho}, H_l \big] \Big\|_2^2 \bigg)^{1/2} }. \nonumber 
\end{align}

By Eq.~\eqref{eq:BoundSingleVariance}, a higher uncertainty $\Delta H_l$ in $H_l = \mu_l A_l$ can enable better precision. Meanwhile, Eq.~\eqref{eq:singleparamBound} constrains the relative error in terms of the classical uncertainty in $H_l = \mu_l A_l$. Equation~\eqref{eq:singleparamBound} also reveals the role of noncommutativity: when $A_l$ does not commute with $\rho$, the ability to estimate $\mu_l$ diminishes. This fact has an analogue in single-parameter estimation in unitary quantum metrology, as detailed in Appendix~\ref{app:nc-and-Hamiltonian-evolution}. There, $\mu_l$ can be encoded in a probe state via Hamiltonian evolution under $H=\mu_lA_l+H'$, for an arbitrary Hermitian $H'$. If $[A_l,H]\neq 0$---and so $[\rho, H] \neq 0$ for thermal states $\rho$---the ability to measure $\mu_l$ is diminished~\footnote{Of course, in unitary quantum metrology, whether $A_l$ commutes with $\sqrt{\rho}$ does not matter: no thermal state will evolve under $H$ and gain information about $\mu_l$.}.

 In quantum metrology, the estimation error's scaling with a sensor's size can constitute an entanglement advantage. Consider a system of $n$ subsystems and $H_l$ a sum of $n$ local terms. Superextensive variances $\big(\Delta H_l \big)^2 \sim n^{\alpha}$, with $\alpha > 1$, are atypical for thermal states
of spatially-local Hamiltonians. 
For instance, $\big(\Delta H_l \big)^2 \sim n$ for states with exponentially decaying correlations~\cite{gong2022bounds,
alhambra2022quantum}. 
From Eq.~\eqref{eq:BoundSingleVariance}, one would expect the optimal estimation error to scale as $1/(\beta \sqrt{\mathcal{N}} \sqrt{n})$, as in the standard quantum limit~\cite{SQL2020,SQL2023}. At critical points, however, $\big(\Delta H_l \big)^2 \sim n$ may be violated~\cite{ZanardiPRA2017b,
FrerotPRL2018,
Gabbrielli_2018_PhaseTransitions}.
We can observe violations also with certain nonlocal Hamiltonians.

We now show that one can beat the standard quantum limit in Hamiltonian metrology using thermal states. Consider estimating a field $\mu$ by measuring copies of a thermal state of the $n$-qubit Hamiltonian $H \nobreak = \nobreak  \mu \sum_{j=1}^n \left( \sigma_z^j+ 1 \right) \nobreak -\nobreak \lambda    \bigotimes_{j=1}^n n \sigma_x^j \nobreak\equiv\nobreak H_\mu \nobreak+\nobreak H_\lambda$. We assume $\lambda > 0$ and $\mu > 0$. 
Let $\ket{\overline{0}}$ denote the $n$-fold tensor product of the eigenvalue-$(-1)$ eigenstate of $\sigma_z$; and $\ket{\overline{1}}$, the product of the eigenvalue-$1$ eigenstate. The $n$-qubit GHZ state $\ket{\Phi} \nobreak \coloneqq \nobreak  \left(  \ket{ \overline{0 } } \nobreak + \nobreak  \ket{ \overline{1 } } \right)/\sqrt{2}$ is a ground state of $H_\lambda$.
We prove in  Appendix~\ref{sec-app:SQL} that  $\ket{\Phi}$ is the unique ground state if $H_\mu$ is a perturbation ($\mu/\lambda \ll 1$). The variance of $H_\mu$ in $\ket{\Phi}$ is $\langle \Phi \vert H_\mu^2 \ket{\Phi} \nobreak - \langle \Phi \vert H_\mu\ket{\Phi}^2\nobreak= \nobreak \mu^2 n^2$. Therefore, one might expect that $ \Delta H_\mu \sim \mu n^\alpha$, with $\alpha > 1/2$, in low-temperature thermal states. In Appendix~\ref{sec-app:SQL}, we prove this expectation, showing that $\alpha=1$ for  $\beta \lambda n \gg 1$. Note this proof does not require that $\mu/\lambda \ll 1$. 
By Eqs.~\eqref{eq:BoundSingleVariance} and~\eqref{eq:singleparamBound}, this result suggests a minimum relative estimation error that decreases faster than the standard quantum limit $1/\sqrt{n}$.  Figure~\ref{fig:fig2} supports this argument, exhibiting a regime with optimal relative estimation errors below~$1/\sqrt{n}$. These results would have been difficult to deduce from the expression~\eqref{eq:FisherMatrix} for the quantum Fisher information. By leveraging our bounds, we found a model that beats the standard quantum limit.

\begin{figure}
Relative estimation error vs. number of qubits \\
  \centering  
  \includegraphics[trim=00 00 00 00,width=0.5\textwidth]{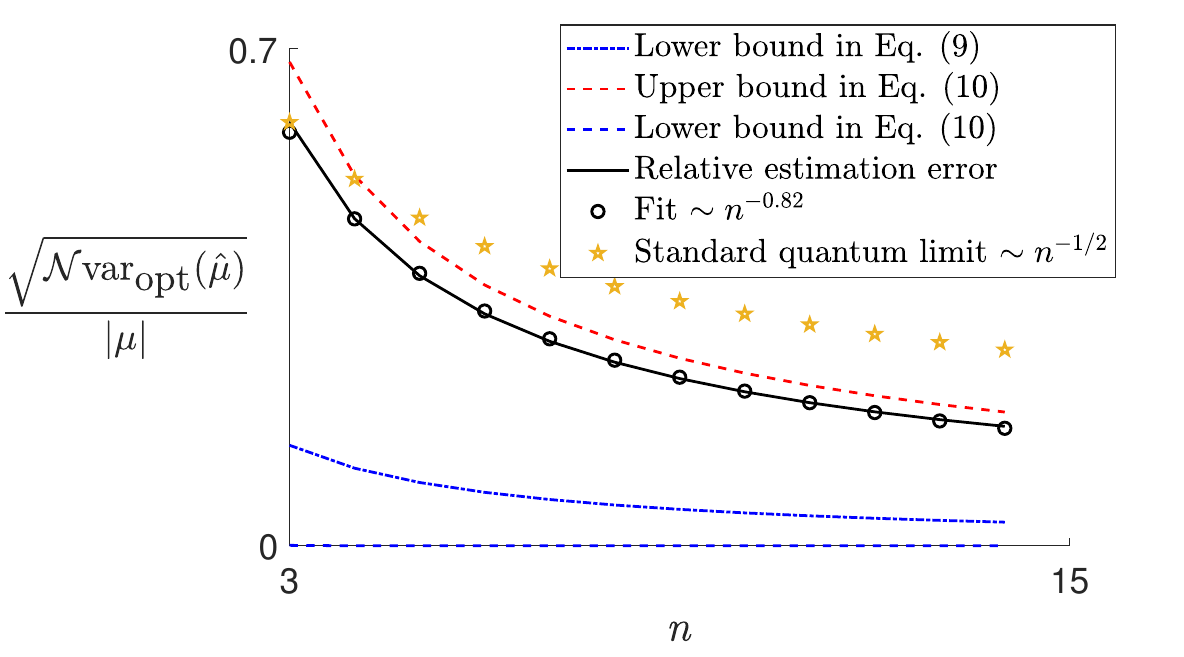} 
\caption{ 
\textbf{Beating the standard quantum limit.} 
The figure shows the relative estimation error $\sqrt{\var_{\textnormal{opt}}(\hat{\mu})}/|\mu|$ for the parameter $\mu$ in the  $n$-qubit 
$H \nobreak = \nobreak  \mu \sum_j^n \left( \sigma_z^j + 1 \right) \nobreak -\nobreak  \lambda  \bigotimes_j^n  n \sigma_x^j \nobreak\coloneqq\nobreak H_\mu \nobreak+\nobreak H_\lambda$. The bounds appear in Eqs.~\eqref{eq:BoundSingleVariance} and~\eqref{eq:singleparamBound} [the upper bound in~\eqref{eq:BoundSingleVariance} is, here, too loose to appear in the plotted range]. We take $\lambda \beta  = 2 \mu \beta = 6$. As we show in  Appendix~\ref{sec-app:SQL}, $ \Delta H_\mu \sim \mu n$, for large $\beta \lambda n$. A consequence, suggested by Eqs.~\eqref{eq:BoundSingleVariance} and~\eqref{eq:singleparamBound}, is an optimal estimation error that decays faster than $1/\sqrt{n}$.
\label{fig:fig2}
}
\end{figure}

\vspace{7pt}
\textit{Bounds on multiparameter estimation errors.---}The single-parameter bounds above apply when all parameters except the target parameter are known. However, our results imply bounds on the error in joint estimates of $M$ Hamiltonian parameters. 
The variances' sum serves as the error measure. We aim for a total error $\sum_{l=1}^M \textnormal{var}(\hat\mu_l) \, = \,  \epsilon_{\textnormal{err}}^2$.  
By the multiparameter Cram\'er-Rao bound~\eqref{eq:multiCR},
\begin{align}
\label{eq:boundMulti}
\epsilon_{\textnormal{err}}^2 = \sum_{l=1}^M \textnormal{var}(\hat\mu_l) \geq  \frac{1}{\mathcal{N}} \tr{ \info^{-1} } \geq  \frac{1}{\mathcal{N}} \sum_{l=1}^M \frac{1}{\info_{ll}} . 
\end{align}
The final inequality holds under the condition $\info > 0$, satisfied if one can estimate every linear combination of parameters~\cite{liu2019quantum}.
The second inequality is useful for large $M$, when calculating  $\info^{-1}$ is computationally hard.

The second inequality is saturated if and only if $\info$ is diagonal. The first inequality is saturated if and only if
\begin{align}
\label{eq:saturability}
   0 &= \tr{\rho [L_l,L_m]} \\
   &=4 \sum_{\omega_j \neq \omega_k} \frac{(p_j - p_k)^3}{(\omega_j - \omega_k)^2 (p_j+p_k)^2} \bra{j} A_l \ket{k} \bra{k} A_m \ket{j} \nonumber 
\end{align}
for all $\{l,m\}$, where $p_j = e^{-\beta \omega_j}/Z_\beta$. These are rather stringent conditions violated by typical many-body Hamiltonians. 

By combining Eq.~\eqref{eq:boundMulti} with the bounds~\eqref{eq:UpperBoundsDiagonalDeltaH} and~\eqref{eq:BoundsDiagonalWYupper}, we bound the error in the estimation of multiple Hamiltonian parameters. 
To learn $M$ Hamiltonian parameters with an error 
$\epsilon_{\textnormal{err}}$, one needs a number
$\mathcal{N}$ of measurements satisfying
\begin{subequations}
\label{eq:MultiBound}
\begin{align}
\mathcal{N} &\geq \frac{ 1 }{ \beta^2 \, \epsilon_{\textnormal{err}}^2} \, \sum_{l=1}^M \frac{ c_2^{-1}/2 }{  \left( \Delta A_l \right)^2 - \tfrac{1}{2}  \big\| \left[\sqrt{\rho}, A_l \right] \big\|_2^2 }
\quad \text{and} \\
\mathcal{N} 
&\geq  \frac{ 1 }{ \beta^2 \, \epsilon_{\textnormal{err}}^2} \, \sum_{l=1}^M \frac{1}{ \left( \Delta A_l \right)^2 }.
\end{align}
\end{subequations}
Consequently, 
\begin{align}
\label{eq:BoundMultiple}
\mathcal{N} = \Omega \left( \frac{ M }{ \beta^2 \, \epsilon_{\textnormal{err}}^2} \min_l \frac{1}{\big( \Delta A_l \big)^2}\right).
\end{align}

We can compare Eq.~\eqref{eq:BoundMultiple} to complexity-theoretic results~\cite{anshu2021sample,haah2022optimal} about the number $\mathcal{N}$ of copies of the state required to learn $M$ Hamiltonian parameters to within an \emph{l}$_2$-distance error $\epsilon$ defined through $\epsilon^2 \nobreak = \nobreak \sum_{l=1}^M (\hat\mu_l - \mu_l )^2$. 
At least
$
\mathcal{N} = \Omega \left( \frac{\exp(\beta) M}{ \beta^2 \, \epsilon^2} \right)
$
samples are required for an $M$-qubit Hamiltonian~\cite{haah2022optimal}.
At low temperatures, their bound is tighter, as a function of $\beta$. 
Moreover, we have only proven Eq.~\eqref{eq:MultiBound} to be saturable under stringent conditions on the operators $A_l$. They prove a stronger result: $\mathcal{N} = \mathcal{O}\left( \frac{M}{ \beta^2 \, \epsilon^2} \ln\left( M/\delta \right)\right)$ samples suffice to learn the parameters with a constant failure probability $\delta$.
In contrast, our results are more general since they concern the average error in estimations of parameters in arbitrary Hamiltonians. Also, our results reveal the roles of uncertainties $\Delta A_l$, and of the state's noncommutativity with $A_l$, in the estimation error. 
We compare this Letter's bounds with previous bounds in detail in Appendix~\ref{sec-app:HamiltonianLearning} (see Table~\ref{table:complexity_Ham_learning}).

\vspace{7pt}
\textit{Estimation of chemical potentials.---}In the presence of conserved charges $Q_l$, thermalizing systems reach generalized Gibbs states~\cite{GGERigol2016,
GGENicole2016,
GGEGuryanova2016,
GGEmonnai2020}
\begin{align}
\rho_{\beta,\{\mu_l\}} 
= e^{-\beta \left( H_0 + \sum_l \mu_l Q_l\right)} / Z_{\beta,\{\mu_l\}}.
\end{align}
 $H_0$ is the system Hamiltonian. The $\mu_l$ are the \emph{chemical potentials} corresponding to the charges, which satisfy $[H_0,Q_l] = 0$ for all $l$. 

Our results imply constraints on the minimum error in estimations of the chemical potentials: we identify $H \equiv H_0 + \sum_l \mu_l Q_l$ and $A_l \equiv Q_l$ in Eq.~\eqref{eq:Gibbs}. For example, consider estimating one $\mu_l$. Equations~\eqref{eq:BoundsDiagonalWY}, with the quantum Cram\'er-Rao bound's saturability, imply
\begin{align}
\label{eq:singleparamBoundChem}
&\frac{1}{2.4 \, c_2 \, \beta^2  \mathcal{N} \, \bigg( \left( \Delta Q_l \right)^2 -   \tfrac{1}{2} \Big\| \big[\sqrt{\rho}, Q_l \big] \Big\|_2^2 \bigg) }  \nonumber \\
&\qquad \qquad \leq   \textnormal{var}_{\textnormal{opt}} ( \hat\mu_l )   \leq  \\ 
&\frac{1}{0.8 \beta^2 \mathcal{N} \, \bigg(  \left( \Delta Q_l \right)^2 - \tfrac{1}{2}  \Big\| \big[\sqrt{\rho}, Q_l \big] \Big\|_2^2 \bigg) }. \nonumber
\end{align}

Classically, all charges commute with each other and so with $\rho$. 
Quantum charges can defy this expectation: $[Q_l,Q_m] \neq 0$~\cite{Lostaglio_2017,
GGEGuryanova2016,
GGENicole2016,
GGELandiPRXQ2022,GGENicole2023}. 
For instance, the two-qubit Hamiltonian $H_0 \nobreak=\nobreak \sigma_z\otimes\sigma_z$ conserves charges $Q_1 \nobreak=\nobreak \sigma_z \otimes \id$ and $Q_2 \nobreak=\nobreak \sigma_x\otimes\sigma_x$ that do not commute with each other. This noncommutation prevents charges from commuting with the state: $[Q_l,\sqrt{\rho}] \neq 0$. 
This lack of equality implies a quantum disadvantage in parameter estimation: charges' noncommutativity hinders the ability to measure chemical potential $\mu_1$.

\vspace{7pt}
\textit{Discussion.---} 
Our bounds highlight how estimation error depends on the noncommutativity of the operators defining the Hamiltonian. The noncommutativity engenders a disadvantage, diminishing precision. See Eq.~\eqref{eq:singleparamBound} and Appendix~\ref{app:nc-and-Hamiltonian-evolution} for a comparison with the estimation of parameters from Hamiltonian evolution. 

Furthermore, we found that noncommutativity of conserved charges hinders estimations of chemical potentials.
This result contrasts with Refs.~\cite{GGELandiPRXQ2022,Non-AbelianAdvantages2023}, which show that conserved quantities' noncommutativity provides an advantage in quantum transport processes by decreasing entropy production. Our work therefore contributes to the debate about whether noncommuting charges enhance or hinder desirable properties in information-processing and thermodynamic tasks~\cite{Majidy_23_Non,Majidy_23_Noncommuting}.

A natural open problem, unexplored in this work, is the construction of concrete protocols that saturate the bounds~\cite{len2022quantum,yin2023heisenberglimited}. 
 Moreover, we found a toy model where, using measurements on a thermal state, one can beat the standard quantum limit for the task of estimating (a component of) a field coupled locally to $n$ qubits. Further work could shed light on whether one can use thermal states also of more-physically-realistic, fully local Hamiltonians to beat the standard quantum limit, possibly by exploiting criticality~\cite{ZanardiPRA2017b,
FrerotPRL2018,
Gabbrielli_2018_PhaseTransitions,yu2023criticality}.

\vspace{7pt}
\textit{Note added.---}Ref.~\cite{bakshi2023learning}, which studies the Hamiltonian learning problem at all temperatures, was posted during the preparation of this manuscript.

\textit{Acknowledgments.---}L.P.G.P.~thanks \'Alvaro Alhambra for discussions and Chris White for comments on the draft. The work at Los Alamos National Laboratory (LANL) was carried out under the auspices of the US DOE and NNSA under contract No.~DEAC52-06NA25396. 
L.P.G.P. also acknowledges support by the DOE Office of Science, Office of Advanced Scientific Computing Research, Accelerated Research for Quantum Computing program, Fundamental Algorithmic Research for Quantum Computing (FAR-QC) project, Laboratory Directed Research and Development (LDRD) program of LANL under project number 20230049DR, and Beyond Moore’s Law project of the Advanced Simulation and Computing Program at LANL managed by Triad National Security, LLC, for the National Nuclear Security Administration of the U.S. DOE under contract 89233218CNA000001.   J.B., A.E., and A.V.G~were supported in part by the DoE ASCR Accelerated Research in Quantum Computing program (award No.~DE-SC0020312), NSF QLCI (award No.~OMA-2120757), AFOSR, DoE ASCR Quantum Testbed Pathfinder program (awards No.~DE-SC0019040 and No.~DE-SC0024220), NSF PFCQC program, ARO MURI, AFOSR MURI, and DARPA SAVaNT ADVENT. Support is also acknowledged from the U.S.~Department of Energy, Office of Science, National Quantum Information Science Research Centers, Quantum Systems Accelerator.   K.B.~thanks Dax Koh, Naresh Goud Boddu and Rahul Jain for interesting discussions. H.D.~acknowledges the Simons Foundation ``Simons Investigator'' program for Mathematical and Physical Sciences, the U.S.~Department of Energy Quantum Systems Accelerator (QSA) Research Center (DE-FOA-0002253), and the National Science Foundation QLCI grant for Robust Quantum Systems (OMA-2120757). 
This work received support from the National Science Foundation (QLCI grant OMA-2120757), and the John Templeton Foundation (award no.~62422). Specific
product citations are for the purpose of clarification only, and
are not an endorsement by the authors or by NIST.

\bibliography{references}

\clearpage
\newpage
\onecolumngrid

\section*{Appendix}
\setcounter{secnumdepth}{1}
\setcounter{equation}{0}
\renewcommand{\theequation}{A\arabic{equation}}

Appendix~\ref{sec-app:FisherInfo} --- Derivation of the quantum Fisher information matrix of Hamiltonian parameters for thermal states.\vspace{7pt}

 Appendix~\ref{sec-app:BoundsOnFisher} --- Upper and lower bounds on the diagonals of the quantum Fisher information matrix: \\
proof of Eqs.~\eqref{eq:BoundsDiagonalDeltaH} 
and~\eqref{eq:BoundsDiagonalWY} in the main text.
\vspace{7pt}

Appendix~\ref{app:nc-and-Hamiltonian-evolution} --- Role of noncommutativity in parameter estimation.
\vspace{7pt}

 Appendix~\ref{sec-app:QubitExample} --- Quantum Fisher information of a two-level system.
\vspace{7pt}

Appendix~\ref{sec-app:Comparisons} --- Comparisons of the bounds on the quantum Fisher information.
\vspace{7pt}

Appendix~\ref{sec-app:SQL} --- Properties of the Hamiltonian $H = \mu\sum_{j=1}^n \left( \sigma_z^j+ 1 \right) \nobreak -\nobreak \lambda    \bigotimes_{j=1}^n n \sigma_x^j$.
\vspace{7pt}

 Appendix~\ref{sec-app:saturability} --- Conditions for saturability of the multiparameter Cram\'er-Rao bound: \\
proof of Eq.~\eqref{eq:saturability} in the main text.
\vspace{7pt}

 Appendix~\ref{sec-app:HamiltonianLearning} --- Comparisons with the literature on Hamiltonian learning.
 \vspace{7pt}

\section{The quantum Fisher information matrix}
\label{sec-app:FisherInfo}
In this section, we derive the closed-form expression for the quantum Fisher information matrix of Hamiltonian parameters for thermal states.

The quantum Fisher information matrix has elements
\begin{align}
\label{eq-app:QFIMatrix}
    \info_{lm} \coloneqq 2 \sum_{jk} \frac{ \textnormal{Re} \left[ \bra{j} \partial_l \rho \ket{k}\!  \bra{k} \partial_m \rho \ket{j} \right] }{p_j + p_k} \, .
\end{align}
We have defined $\partial_l \coloneqq \tfrac{\partial}{\partial \mu_l}$.
The matrix characterizes the precision with which multiple parameters $\mu_l$ can be estimated. Let $\mathcal{N}$ denote the number of measurements performed. The multiparameter Cram\'er-Rao bound says that~\cite{liu2019quantum}
\begin{align}
    \label{eq-app:multiCR}
    \cov(\hat{\vec{\mu}}) \geq \frac{1}{\mathcal{N}}\info^{-1} \, .
\end{align}
This bound is asymptotically saturable if and only if
\begin{align}
    \tr{\rho [L_l,L_m]}
= 0 .
\end{align}
The symmetric logarithmic derivative is defined by $\partial_l \rho = \tfrac{1}{2}\{ L_l,\rho\}$. 

Throughout this appendix, we omit the temperature dependence from the partition-function notation: $Z \equiv Z_\beta$.
Since $\rho = e^{-\beta H}/Z = \sum_j e^{-\beta \omega_j}\ket{j}\!\bra{j}/Z$, the derivative in Eq.~\eqref{eq-app:QFIMatrix} is
\begin{align}
\label{eq-app:derivrho}
\partial_l \rho &= \frac{1}{Z} \partial_l e^{-\beta H} - \rho \frac{ \partial_l Z }{Z}  =  \frac{1}{Z} \left[  \partial_l e^{-\beta H}  - \rho \tr{\partial_l e^{-\beta H}} \right].
\end{align}
We must calculate the matrix elements of $\partial_l e^{-\beta H}$. 
 Using the Taylor series
\begin{align}
e^{-\beta H} = \sum_{n=0}^\infty (-\beta)^n \frac{H^n}{n!} \, ,
\end{align}
we obtain
\begin{align}
\label{eq-app:offdiagderivrho}
\bra{j} \partial_l e^{-\beta H} \ket{k} &=  \sum_{n=0}^\infty \frac{(-\beta)^n}{n!} \bra{j} \partial_l H^n \ket{k} \nonumber \\
&=  \sum_{n=1}^\infty \frac{(-\beta)^n}{n!} \bra{j} \sum_{m = 0}^{n-1} H^m A_l H^{n-m-1} \ket{k} \nonumber \\
&=  \sum_{n=1}^\infty \frac{(-\beta)^n}{n!}   \sum_{m = 0}^{n-1} \omega_j^m  \omega_k^{n-m-1} \bra{j} A_l \ket{k} \nonumber \\
&= \bra{j} A_l \ket{k} \sum_{n=1}^\infty \frac{(-\beta)^n}{n!}   \sum_{m = 0}^{n-1} \omega_j^m  \omega_k^{n-m-1} \nonumber \\
&\eqqcolon \bra{j} A_l \ket{k} \Gamma_{jk} . 
\end{align}
We have defined 
\begin{align}
\Gamma_{jk} \coloneqq \sum_{n=1}^\infty \frac{(-\beta)^n}{n!}   \sum_{m = 0}^{n-1} \omega_j^m  \omega_k^{n-m-1}
\end{align}
as a function of the temperature and of the Hamiltonian's spectrum.

We can re-express $\Gamma_{jk}$ using the formula for an infinite geometric series: if $\omega_j \neq \omega_k$, then
\begin{align}
\label{eq-app:gammaoffdiag}
 \Gamma_{jk} & = \sum_{n=1}^\infty \frac{(-\beta)^n}{n!}  \frac{\omega_j^{n}-\omega_k^{n}}{\omega_j-\omega_k} =\frac{e^{-\beta \omega_j} - e^{-\beta \omega_k}}{\omega_j-\omega_k} = Z \frac{\left( p_j  - p_k \right)}{\omega_j - \omega_k} 
   \, , \qquad \text{for \quad }\omega_j \neq \omega_k .
\end{align}
If $\omega_j = \omega_k$, then
\begin{align}
\label{eq-app:gammadiag}
\Gamma_{jk}  &= \sum_{n=1}^\infty \frac{(-\beta)^n}{n!}   \sum_{m = 0}^{n-1} \omega_j^{n-1} = \sum_{n=1}^\infty \frac{(-\beta)^n}{n!} n \omega_j^{n-1} \nonumber \\
&=  - \beta \sum_{n=1}^\infty \frac{(-\beta)^{n-1}}{\left(n-1\right)!}  \omega_j^{n-1} = - \beta e^{-\beta \omega_j} \nonumber \\
&= -  \beta Z p_j 
\, , \qquad \text{for \quad }\omega_j = \omega_k. 
\end{align}
Using Eqs.~\eqref{eq-app:gammadiag} and~\eqref{eq-app:offdiagderivrho} we can evaluate the first term in Eq.~\eqref{eq-app:derivrho}:
\begin{align}
\label{eq-app:traceaux}
\frac{\tr{\partial_l e^{-\beta H}}}{Z} &= \frac{1}{Z} \sum_j \bra{j} A_l \ket{j} \Gamma_{jj}  = - \beta \frac{1}{Z} \sum_j \bra{j} A_l \ket{j} e^{-\beta \omega_j}   = - \beta \langle A_l \rangle.
\end{align}
We denote thermal averages by $\langle A_l \rangle \coloneqq \tr{A_l \rho}$.

Substituting from Eq.~(\ref{eq-app:traceaux}) into Eq.~(\ref{eq-app:derivrho}) yields
\begin{align}
\partial_l \rho &=  \frac{1}{Z}  \partial_l e^{-\beta H}  + \beta \langle A_l \rangle \rho .
\end{align}
Furthermore, substituting into Eq.~\eqref{eq-app:offdiagderivrho} from Eqs.~(\ref{eq-app:gammaoffdiag}) and~(\ref{eq-app:gammadiag}) yields
\begin{align}
\bra{j} \partial_l e^{-\beta H} \ket{k} &=\bra{j} A_l \ket{k} Z \frac{\left( p_j  - p_k \right)}{\omega_j - \omega_k} \, , \qquad \text{for \quad }\omega_j \neq \omega_k,
\end{align}
 and 
\begin{align}
\bra{j} \partial_l e^{-\beta H} \ket{k} &=\bra{j} A_l \ket{k} \Gamma_{jj} =- \bra{j} A_l \ket{k} \beta Z p_j \qquad \text{for \quad } \omega_j = \omega_k.
\end{align}

Let $\delta A_l \coloneqq A_l - \langle A_l \rangle = A_l - \tr{\rho A_l}$. If $\omega_j \neq \omega_k$, then
\begin{align}
\label{eq-app:statederivOffDiag}
    \bra{j} \partial_l \rho \ket{k} &= \bra{j} A_l \ket{k} \frac{\left( p_j  - p_k \right)}{\omega_j - \omega_k} +  \beta \langle A_l \rangle \bra{j} \rho \ket{k}
    = \bra{j} \delta A_l \ket{k} \frac{\left( p_j  - p_k \right)}{\omega_j - \omega_k} \, ,
      \qquad \textnormal{for } \omega_j \neq \omega_k,
\end{align}
whereas, if $\omega_j = \omega_k$,
\begin{align}
\label{eq-app:statederivDiag}
    \bra{j} \partial_l \rho \ket{k} &=  -  \bra{j} A_l \ket{k} \beta  p_j  + \beta \langle A_l \rangle p_j \delta_{jk} 
     =  -  \bra{j} \delta A_l \ket{k} \beta  p_j =  -  \bra{j} \delta A_l \ket{k} \beta  \frac{p_j+p_k}{2} \, , \qquad \textnormal{for } \omega_j = \omega_k.
\end{align}
 
Thus, the quantum Fisher information matrix in Eq.~\eqref{eq-app:QFIMatrix} becomes
\begin{align}
    \info_{lm} &\coloneqq 2 \sum_{j k} \frac{ \textnormal{Re} \left[ \bra{j} \partial_l \rho \ket{k}\!  \bra{k} \partial_m \rho \ket{j} \right] }{p_j + p_k} \nonumber \\
    &=  2 \sum_{\omega_j \neq \omega_k} \frac{ (p_j-p_k)^2}{(p_j + p_k)(\omega_j-\omega_k)^2} \textnormal{Re} \left[ \delta A_{jk}^l  \delta A_{kj}^m \right] +   \sum_{\omega_j = \omega_k} \beta^2 \frac{p_j+p_k}{2} \textnormal{Re} \left[ \delta A_{jk}^l  \delta A_{kj}^m \right] \nonumber \\
     &=2\beta^2\sum_{\omega_j \neq \omega_k} p_j \frac{ (1-p_k/p_j)^2}{(1+ p_k/p_j)\ln^2(p_k/p_j)} \textnormal{Re} \left[ \delta A_{jk}^l \delta A_{kj}^m \right]  +  \beta^2  \sum_{\omega_j = \omega_k} \frac{p_j+p_k}{2} \textnormal{Re} \left[ \delta A_{jk}^l  \delta A_{kj}^m \right]. \label{eq_Flm_1}
\end{align}
In Appendix \ref{sec-app:BoundsOnFisher}, we use this expression to upper- and lower-bound $\info_{ll}$.

\section{Bounds on the quantum Fisher information}
\label{sec-app:BoundsOnFisher}

In this section, we upper- and lower-bound the diagonals of the quantum Fisher information matrix. That is, we prove Eqs.~\eqref{eq:BoundsDiagonalDeltaH} 
and~\eqref{eq:BoundsDiagonalWY} from the main text. 
By Eq.~\eqref{eq_Flm_1}, the quantum Fisher information about a parameter $\mu_l$ is
\begin{align}
\label{eq-app:FisherSingle}
    \info_{ll}   
    &=2\beta^2 \sum_{\omega_j \neq \omega_k}  p_j \frac{ (1-p_k/p_j)^2}{(1+ p_k/p_j)\ln^2(p_k/p_j)}  \big| \delta A_{jk}^l \big|^2  + \beta^2  \sum_{\omega_j = \omega_k} \frac{p_j+p_k}{2} \big| \delta A_{jk}^l \big|^2.
\end{align}

\subsection{Upper bound in terms of $(\Delta A_l)^2$}

If $x \coloneqq p_k/p_j$, the first term in the quantum Fisher information [Eq.~\eqref{eq-app:FisherSingle}] depends on $\frac{(1-x)^2}{(1+x)\ln^2(x)}$. It will be convenient to upper-bound this fraction as $\frac{(1-x)^2}{(1+x)\ln^2(x)} \leq (1+x) c_1$, for some $c_1$ to be determined. 
Shifting the $(1+x)$ from the inequality's right-hand side to the left-hand side, we form a fraction 
$\frac{(1-x)^2}{(1+x)^2\ln^2(x)}$ 
that is maximized at $x = 1$. Furthermore, $p_k/p_j$ comes closest to $1$ for energy eigenstates whose energies are as close as possible:
$x_\M \coloneqq e^{-\beta \min_{j,k} \{\omega_k - \omega_j\} } \equiv e^{-\beta \gm }$. We have defined $\gm \coloneqq \min_{j,k} \{\omega_j - \omega_k\}$ as the Hamiltonian's minimum energy gap. 
Combining these observations, we choose 
\begin{align}
c_1(\gm) \coloneqq \frac{(1-x_\M)^2}{(1+x_\M)^2 \ln^2(x_\M)} = \frac{(1-e^{-\beta \gm})^2}{(1+e^{-\beta \gm})^2}\frac{1}{\beta^2 \gm^2} = \frac{\tanh^2(\beta \gm/2)}{\beta^2 \gm^2} \, .
\end{align}
The limiting values of $c_1$, as a function of temperature, are
\begin{align}
    c_1(\gm) \approx
    \begin{cases} & \frac{1}{\beta^2 \gm^2} 
    \qquad \textnormal{for} \qquad \beta \gm \gg 1 \\
     & \frac{1}{4} \qquad \textnormal{for} \qquad \beta \gm \ll 1.
     \end{cases}
\end{align}
Applying this choice and the general bound above to Eq.~\eqref{eq-app:FisherSingle}, we bound the quantum Fisher information about a parameter $\mu_l$:
\begin{subequations}
\begin{align}
    \info_{ll} &=2\sum_{\omega_j \neq \omega_k} \beta^2 p_j \frac{ (1-p_k/p_j)^2}{(1+ p_k/p_j)\ln^2(p_k/p_j)}  \big| \delta A_{jk}^l \big|^2   + \beta^2  \sum_{\omega_j = \omega_k}  \frac{p_j+p_k}{2} \big| \delta A_{jk}^l \big|^2 \nonumber \\
    &\leq  2 c_1(\gm) \beta^2 \sum_{\omega_j \neq \omega_k}  p_j  \left( 1+\frac{p_k}{p_j} \right) \left|   \delta A_{jk}^l   \right|^2  + \beta^2  \sum_{\omega_j = \omega_k} \frac{p_j+p_k}{2} \big| \delta A_{jk}^l \big|^2 \label{eqq-app:QFIAux1} \\ 
    &=  4 c_1(\gm)  \beta^2 \sum_{\omega_j \neq \omega_k}  p_j  \left|   \delta A_{jk}^l   \right|^2  + 4 c_1(\gm)  \beta^2 \sum_{\omega_j = \omega_k}  p_j \big| \delta A_{jk}^l \big|^2 + [1-4 c_1(\gm) ] \beta^2 \sum_{\omega_j = \omega_k}  p_j \big| \delta A_{jk}^l \big|^2  \nonumber \\
    &= 4 c_1(\gm)  \beta^2 \tr{ \rho \left[ \delta A_l \right]^2} +  [1-4 c_1(\gm) ] \beta^2 \sum_{\omega_j = \omega_k}  p_j \big| \delta A_{jk}^l \big|^2 \nonumber \\
    &=\,  4 c_1(\gm)  \beta^2 \left( \Delta A_l \right)^2 
    +  [1-4 c_1(\gm) ] \beta^2 \left( \Delta A_l^D \right)^2. \label{eqq-app:QFIAux2} 
\end{align}
\end{subequations}
We have defined $\Delta A = \sqrt{\langle A^2 \rangle - \langle A \rangle^2}$ as the standard deviation of an operator $A$ in the thermal state. Also, $A_l^D \nobreak\coloneqq\nobreak \sum_{\omega_j = \omega_k} \bra{j}A_l\ket{k} \ket{k}\!\bra{j}$ is the sum of the block-diagonal elements of the matrix that represents $A_l$ relative to the energy eigenbasis. Since $0 \leq c_1(\gm)  \leq 1/4$ in Eq.~\eqref{eqq-app:QFIAux2}, also, 
\begin{align}
\label{eqq-app:QFIAux3}
    \info_{ll} \leq \beta^2 \left( \Delta A_l \right)^2.
\end{align}
We have proved Eq.~\eqref{eq:UpperBoundsDiagonalDeltaH} in the main text.
Bounds~\eqref{eqq-app:QFIAux2} and~\eqref{eqq-app:QFIAux3} are saturated if $A_l$ is diagonal relative to the energy eigenbasis.

\subsection{Lower bound in terms of $(\Delta A_l)^2$}

A similar derivation implies a lower bound on $\info_{ll}$. The function $\frac{(1-x)^2}{(1+x)^2\ln^2(x)}$ is minimized at $x = 0$ and in the limit as $x \rightarrow \infty$. Moreover, $x$ has a minimum value of $x_\m \coloneqq e^{-\beta \gM}$, where $\gM \coloneqq \max_j \omega_j - \min_j \omega_j$, and a maximum value of $x_\M \coloneqq e^{\beta \gM}$. Since $c_1(-\gM) = c_1(\gM)$, 
\begin{align}
\label{eqq-app:QFIAux4}
\frac{(1-x)^2}{(1+x)\ln^2(x)} \geq (1+x) c_1(\gM).
\end{align}
Using Eqs.~\eqref{eqq-app:QFIAux4} and~\eqref{eq-app:FisherSingle} leads to 
\begin{align}
    \info_{ll} &=2\sum_{\omega_j \neq \omega_k} \beta^2 p_j \frac{ (1-p_k/p_j)^2}{(1+ p_k/p_j)\ln^2(p_k/p_j)}  \big| \delta A_{jk}^l \big|^2   +   \beta^2  \sum_{\omega_j = \omega_k}  \frac{p_j+p_k}{2} \big| \delta A_{jk}^l \big|^2 \nonumber \\
    &\geq  2 c_1(\gM)  \sum_{\omega_j \neq \omega_k}  \beta^2 p_j  \left( 1+\frac{p_k}{p_j} \right) \left|   \delta A_{jk}^l   \right|^2  +  \beta^2  \sum_{\omega_j = \omega_k}  \frac{p_j+p_k}{2} \big| \delta A_{jk}^l \big|^2 \nonumber \\ 
    &=  4 c_1(\gM)  \beta^2 \sum_{\omega_j \neq \omega_k}  p_j  \left|   \delta A_{jk}^l   \right|^2  + 4 c_1(\gM)  \beta^2  \sum_{\omega_j = \omega_k}  \frac{p_j+p_k}{2} \big| \delta A_{jk}^l \big|^2 + [1-4 c_1(\gM) ] \beta^2  \sum_{\omega_j = \omega_k}  \frac{p_j+p_k}{2} \big| \delta A_{jk}^l \big|^2 \nonumber \\
    &= 4 c_1(\gM)  \beta^2 \tr{ \rho \left( \delta A_l \right)^2} +  [1-4 c_1(\gM) ] \beta^2 \sum_{j}  p_j \big| \delta A_{jk}^l \big|^2 \nonumber \\
    &=\,  4 c_1(\gM)  \beta^2 \left( \Delta A_l \right)^2.
\end{align}
We have proved Eq.~\eqref{eq:LowerBoundsDiagonalDeltaH} in the main text.

\subsection{Upper bound in terms of $(\Delta A_l)^2 - \tfrac{1}{2}\| [\sqrt{\rho},A_l]\|_2^2$}

We can obtain a distinct upper bound that depends on the Wigner-Yanase skew information. Beginning with Eq.~\eqref{eq-app:FisherSingle}, we split the sum over $\omega_j \neq \omega_k$ into $\omega_j < \omega_k$ and $\omega_j > \omega_k$ terms. We can then collapse terms 
due to the symmetry with respect to the interchange $p_j \leftrightarrow p_k$:
\begin{align}
    \label{eq-app:Faux2}\info_{ll} &=2\sum_{\omega_j \neq \omega_k} \beta^2 p_j \frac{ (1-p_k/p_j)^2}{(1+ p_k/p_j)\ln^2(p_k/p_j)}  \big| \delta A_{jk}^l \big|^2   +   \beta^2  \sum_{\omega_j = \omega_k}  \frac{p_j+p_k}{2} \big| \delta A_{jk}^l \big|^2 \nonumber \\
    &= 2\sum_{\omega_j > \omega_k} \beta^2 p_j \frac{ (1-p_k/p_j)^2}{(1+ p_k/p_j)\ln^2(p_k/p_j)}  \big| \delta A_{jk}^l \big|^2 + 2\sum_{\omega_j < \omega_k} \beta^2 p_j \frac{ (1-p_k/p_j)^2}{(1+ p_k/p_j)\ln^2(p_k/p_j)}  \big| \delta A_{jk}^l \big|^2   +   \beta^2  \sum_{\omega_j = \omega_k}  \frac{p_j+p_k}{2} \big| \delta A_{jk}^l \big|^2 \nonumber \\
     &=4\sum_{\omega_j < \omega_k} \beta^2 p_j \frac{ (1-p_k/p_j)^2}{(1+ p_k/p_j)\ln^2(p_k/p_j)}  \big| \delta A_{jk}^l \big|^2   +  \beta^2  \sum_{\omega_j = \omega_k}  p_j \big| \delta A_{jk}^l \big|^2.
\end{align}
Assume that the energies $\omega_j$ are in ascending order, such that $x_\m \leq x \coloneqq p_k/p_j \leq 1$, for $j < k$.
The first term in~\eqref{eq-app:Faux2} contains a factor of the form $\frac{(1-x)^2}{(1+x)\ln^2(x)}$, which obeys the upper bound
$\frac{(1-x)^2}{(1+x)\ln^2(x)} \nobreak  \leq \nobreak  c_2 \sqrt x $ for $ 0 \leq x \leq 1$, for some $c_2$. 
The minimum value of $x$, at an inverse temperature $\beta$, is $x_\m \coloneqq \min_{\{j,k\}} p_k/p_j = \min_{\{j,k\}} e^{-\beta(\omega_k - \omega_j)} = e^{-\beta \|H\|_s}$. Therefore, 
\begin{align}
\label{eq-app:c2aux}
    c_2 \coloneqq \frac{1}{\beta^2 \| H \|_s^2} e^{\tfrac{1}{2}\beta \|H \|_s} \frac{ \left( 1-e^{-\beta \|H\|_s} \right)^2}{1+ e^{-\beta \| H \|_s}} = \frac{2\sinh(\beta \gM/2) \tanh(\beta \gM/2)}{\beta^2 \gM^2} \geq  0.42.
\end{align}
 The inequality holds because $2\sinh(x/2)\tanh(x/2)/x^2 \geq 0.42$ (as one can check using, e.g., Mathematica).
The limiting values of $c_2$, as a function of temperature, are
\begin{align}
    c_2 \approx 
    \begin{cases} 
    &  e^{\tfrac{1}{2}\beta \|H \|_s} / (\beta^2 \| H \|_s^2) \, ,
    \qquad \textnormal{for} \qquad \beta \| H \|_s \gg 1, \\
    & 1/2 \, , \qquad \textnormal{for} \qquad \beta \| H \|_s \ll 1.
    \end{cases}
\end{align}

Let us apply Eq.~\eqref{eq-app:c2aux}, with the general bound above, to Eq.~\eqref{eq-app:Faux2}:
\begin{subequations}
    \begin{align}
    \info_{ll} &=4\sum_{\omega_j < \omega_k} \beta^2 p_j \frac{ (1-p_k/p_j)^2}{(1+ p_k/p_j)\ln^2(p_k/p_j)}  \big| \delta A_{jk}^l \big|^2   +    \beta^2  \sum_{\omega_j = \omega_k}  p_j \big| \delta A_{jk}^l \big|^2 \nonumber \\
    &\leq 4c_2 \sum_{\omega_j < \omega_k} \beta^2 p_j \sqrt{\frac{p_k}{p_j}}  \big| \delta A_{jk}^l \big|^2   +   \beta^2  \sum_{\omega_j = \omega_k}  p_j \big| \delta A_{jk}^l \big|^2 \nonumber \\
    &= 2c_2 \beta^2 \sum_{\omega_j < \omega_k}  \sqrt{p_j} \sqrt{ p_k}  \big| \delta A_{jk}^l \big|^2 + 2c_2 \beta^2\sum_{\omega_j > \omega_k}  \sqrt{p_j} \sqrt{ p_k}  \big| \delta A_{jk}^l \big|^2   +    \beta^2  \sum_{\omega_j = \omega_k}  p_j \big| \delta A_{jk}^l \big|^2 \nonumber \\
    &\leq 2c_2 \beta^2 \sum_{\omega_j < \omega_k}  \sqrt{p_j} \sqrt{ p_k}  \big| \delta A_{jk}^l \big|^2 + 2c_2 \beta^2\sum_{\omega_j > \omega_k}  \sqrt{p_j} \sqrt{ p_k}  \big| \delta A_{jk}^l \big|^2   +    \frac{c_2}{0.42} \beta^2  \sum_{\omega_j = \omega_k}  p_j \big| \delta A_{jk}^l \big|^2 \label{eq:app-Faux} \\
    &\leq 2.4 c_2 \beta^2 \sum_{\omega_j < \omega_k}  \sqrt{p_j} \sqrt{ p_k}  \big| \delta A_{jk}^l \big|^2 + 2.4 c_2 \beta^2\sum_{\omega_j > \omega_k}  \sqrt{p_j} \sqrt{ p_k}  \big| \delta A_{jk}^l \big|^2   +    2.4 c_2 \beta^2  \sum_{\omega_j = \omega_k}  p_j \big| \delta A_{jk}^l \big|^2 \label{eq:app-Faux2} \\
    &= 2.4 c_2 \, \beta^2 \, \tr{\sqrt{\rho} \delta A_l \sqrt{\rho} \delta A_l}.
\end{align}
\end{subequations}
In Eqs.~\eqref{eq:app-Faux} and~\eqref{eq:app-Faux2}, we invoked $1 \leq c_2/0.42 \leq 2.4 c_2$.
Since $\tr{\sqrt{\rho} \delta A_l \sqrt{\rho} \delta A_l} = \left( \Delta A_l \right)^2 -    \tfrac{1}{2} \big\| \left[\sqrt{\rho}, A_l \right] \big\|_2^2$, we have proved the second upper bound on $\info_{ll}$, Eq.~\eqref{eq:BoundsDiagonalWY} in the main text.

\subsection{Lower bound in terms of $(\Delta A_l)^2 - \tfrac{1}{2}\| [\sqrt{\rho},A_l]\|_2^2$}

Our general expression $\frac{(1-x)^2}{(1+x)\ln^2(x)}$ obeys the upper bound $ \sqrt x / 2.5 \nobreak \leq \nobreak  \frac{(1-x)^2}{(1+x)\ln^2(x)}$. Applying this bound to Eq.~\eqref{eq-app:FisherSingle} yields
\begin{align}
\label{eq-app:proofboundsInfoDiag}
    \info_{ll} &\geq  \frac{2}{2.5}\sum_{\omega_j \neq \omega_k}  \beta^2 p_j  \sqrt{\frac{p_k}{p_j}} \left|   \delta A_{jk}^l   \right|^2  +    \beta^2  \sum_{\omega_j = \omega_k}  p_j \big| \delta A_{jk}^l \big|^2 \nonumber \\
    &\geq   0.8 \, \beta^2 \sum_{\omega_j \neq  \omega_k}  \sqrt{p_j}  \sqrt{p_k} \left|   \delta A_{jk}^l   \right|^2  + 0.8 \,    \beta^2  \sum_{\omega_j = \omega_k}  p_j \big| \delta A_{jk}^l \big|^2 \nonumber \\
    &=0.8 \, \beta^2 \tr{\sqrt{\rho} \delta A_l \sqrt{\rho} \delta A_l}.
\end{align}
This result completes the proof of Eq.~\eqref{eq:BoundsDiagonalWY} in the main text.

\section{Noncommutativity and Parameter Estimation}\label{app:nc-and-Hamiltonian-evolution}

In this section, we discuss the role of noncommutativity in parameter estimation.
In Eq.~\eqref{eq:singleparamBound} of the main text, we presented an upper and a lower bound on the optimal relative estimation error $\sqrt{\var_{\textnormal{opt}}(\hat{\mu_l})}/|\mu_l|$ with which a parameter $\mu_l$ can be estimated from $\mathcal{N}$ copies of a thermal state. We reproduce the bound here for convenience:
\begin{subequations}
\begin{align}\label{eq-app:single-param-bnd}
\frac{1}{\sqrt{2.4c_2} \beta \, \sqrt{\mathcal{N}} \bigg( \left( \Delta H_l \right)^2 -   \tfrac{1}{2} \Big\| \big[\sqrt{\rho}, H_l \big] \Big\|_2^2 \bigg)^{1/2} } \leq \frac{ \sqrt{\textnormal{var}_{\textnormal{opt}} ( \hat\mu_l ) }}{|\mu_l|} \leq \frac{1}{\sqrt{0.8} \beta \, \sqrt{\mathcal{N}} \bigg( \left( \Delta H_l \right)^2 -   \tfrac{1}{2} \Big\| \big[\sqrt{\rho}, H_l \big] \Big\|_2^2 \bigg)^{1/2} } \, .
\end{align}
\end{subequations}
Recall that $H_l$ is the Hamiltonian term that contains the parameter $\mu_l$.
Due to the $\|[\sqrt{\rho}, H_l]\|_2$, noncommutativity between the state and $H_l$ negatively impacts one's ability to estimate $\mu_l$. Here, we elaborate on the role of noncommutativity in estimating a  parameter from Hamiltonian evolution (as opposed to from a thermal state).

In the Hamiltonian-evolution setting, we estimate $\mu_l$ by evolving a probe state under a Hamiltonian 
\begin{equation}
H=H_l+H' \equiv \mu_l A_l +H'
\end{equation}
for some time $t$. $A_l$ is a Hermitian matrix (the generator of translations associated with $\mu_l$). $H'$ contains all the (possibly time-dependent) terms independent of $\mu_l$. In our setting, $H'=\sum_{j\neq l}H_j$. The time-evolved state $\varrho(t)$ depends on $\mu_l$. One can estimate $\mu_l$ from properly chosen measurements of copies of $\varrho(t)$. The minimum achievable variance is bounded in the single-parameter quantum Cram\'{e}r-Rao bound, Eq.~(5) in the main text. 

The minimal variance can be achieved with a pure probe state 
$\varrho=\ket{\psi}\!\bra{\psi}$. We have defined 
$\ket\psi=(\ket{\lambda_\mathrm{max}}+\ket{\lambda_\mathrm{min}})/\sqrt{2}$, $\ket{\lambda_\mathrm{max}}$ and $\ket{\lambda_\mathrm{min}}$ denoting the eigenstates associated with the maximum and minimum $A_l$ eigenvalues, $\lambda_\mathrm{max}$ and $\lambda_\mathrm{min}$~\cite{giovannetti2006quantum}. Suppose that 
$\ket{\lambda_\mathrm{max}}$ and $\ket{\lambda_\mathrm{min}}$ are $H'$ eigenstates associated with unit eigenvalues:
\begin{align}
   \label{eq_eigen_condns}
   H'\ket{\lambda_\mathrm{max}}=\ket{\lambda_\mathrm{max}}
   , \quad \text{and} \quad
   H'\ket{\lambda_\mathrm{min}}=\ket{\lambda_\mathrm{min}} .
\end{align}
Evolution under $H$ yields a final state $\ket{\psi(t)}=(\ket{\lambda_\mathrm{max}}+e^{(\lambda_\mathrm{max}-\lambda_\mathrm{min})\mu_l t}\ket{\lambda_\mathrm{min}})/\sqrt{2}$, from which $\mu_l$ can be extracted with a variance $\sim [t(\lambda_\mathrm{max}-\lambda_\mathrm{min})]^{-2}$, which is optimal~\cite{boixo2008quantum}. 

The conditions~\eqref{eq_eigen_condns}, under which this optimal scheme works, can be replaced with the weaker condition $[H',H_l]=[H,H_l]=0$. From here, we see the connection to Eq.~(\ref{eq-app:single-param-bnd}): for Gibbs states $\rho$, if $[H, H_l]=0$, then $[\sqrt{\rho}, H_l]=0$. Consequently, we see a direct formal connection between the fact that noncommutativity of $H_l$ with $H$ negatively impacts the estimation of $\mu_l$ through Hamiltonian evolution and the fact that $[\sqrt{\rho}, H_l]\neq 0$ negatively impacts Hamiltonian learning from Gibbs states.

\section{Quantum Fisher information of a two-level system}
\label{sec-app:QubitExample}
 
In this section, we calculate the quantum Fisher information about a parameter in a single-qubit Hamiltonian.
Consider the Hamiltonian 
\begin{align}
H = \Omega_x \sigma_x + \Omega_z \sigma_z + \mu \sigma_z \eqqcolon \vec{v}\cdot\vec{\sigma} .
\end{align}
We have defined the vector $\vec{v} = (\Omega_x,\Omega_z+\mu)$ with the norm $v \coloneqq \sqrt{\Omega_x^2  + (\Omega_z + \mu)^2}$, and $\vec{\sigma} = (\sigma_x,\sigma_z)$ is a vector of Pauli matrices. We aim to estimate $\mu$, so $A = \sigma_z$. The thermal state is
\begin{align}
\label{eq-app:qubitAux}
\rho = \frac{e^{-\beta v \,  \vec{v}\cdot\vec{\sigma}/v }}{Z}=\frac{\cosh(\beta v)\id - \sinh(\beta v) \vec{v}\cdot\vec{\sigma}/v}{Z},
\end{align}
where $Z = 2 \cosh(\beta v)$. The Hamiltonian has a seminorm $\|H\|_s = 2v$.

We directly calculate the Wigner-Yanase skew information,
using $\sqrt{\rho} = e^{-\beta H/2}/\sqrt{Z}$, $A = \sigma_z$, and Eq.~\eqref{eq-app:qubitAux}:
\begin{align}
\label{eq-app:qubitAux2}
\tfrac{1}{2}\big\| [\sqrt{\rho},A] \big\|_2^2 
&= \frac{1}{2 Z} \big\| \sinh(\beta v/2)/v [\vec{v}\cdot\vec{\sigma},\sigma_z] \big\|_2^2 = \frac{\sinh^2(\beta v/2)}{2 v^2 Z} \big\| \Omega_x [\sigma_x,\sigma_z]   \big\|_2^2  \nonumber \\
&= \frac{\sinh^2(\beta v/2)}{2  v^2 Z} \big\| -2i \Omega_x \sigma_y  \big\|_2^2  = 2 \frac{\sinh^2(\beta v/2)}{v^2 Z} \tr{ \big[ -i \Omega_x \sigma_y   \big] \big[ i \Omega_x \sigma_y  \big]}  \nonumber \\
&= 4 \frac{\sinh^2(\beta v/2)}{v^2 Z}   \Omega_x^2    
  =   
    4 \frac{\tfrac{1}{2}(\cosh(\beta v)-1)}{v^2 Z}  \Omega_x^2  =        
    2 \frac{ \cosh(\beta v)-1 }{2 v^2  \cosh(\beta v)}   \Omega_x^2  \nonumber \\
    &=  \frac{ 1 - \sech(\beta v) }{v^2}   \Omega_x^2 \, .
\end{align}
The thermal variance in $\sigma_z$ is 
\begin{align}
\label{eq-app:VarQbit}
(\Delta A)^2 &= \tr{\rho} - [\tr{\rho \sigma_z}]^2 
= 1 - \left[ - \frac{1}{v Z} \sinh(\beta v) \tr{\vec{v}\cdot\vec{\sigma} \sigma_z} \right]^2 = 1 - \frac{4 \sinh^2(\beta v)}{v^2 Z^2} (\Omega_z+\mu)^2 \nonumber \\
& = 1- \frac{\tanh^2(\beta v)}{v^2} (\Omega_z+\mu)^2 \, .
\end{align}
Subtracting Eq.~\eqref{eq-app:qubitAux2} from~\eqref{eq-app:VarQbit} yields
\begin{align}
(\Delta A)^2 - \tfrac{1}{2}\big\| [\sqrt{\rho},A] \big\|_2^2 &= 1- \frac{\tanh^2(\beta v)}{v^2} (\Omega_z+\mu)^2 - \frac{ 1 - \sech(\beta v) }{v^2}   \Omega_x^2 \,  .
\end{align}

We can approximate this expression at high and low temperatures.
If the temperature is high ($\beta v \ll 1$), then $\sech(\beta v) \approx 1 - (\beta v)^2/2 $, and $\tanh(\beta v) \approx \beta v$. Therefore,
\begin{align}
\label{eq-app:FisherQubitAux1}
(\Delta A)^2 - \tfrac{1}{2}\big\| [\sqrt{\rho},A] \big\|_2^2 &\approx 1 - \beta^2 (\Omega_z + \mu)^2 - \frac{\beta^2}{2}  \Omega_x^2  = 1 - \frac{\beta^2}{2} (\Omega_z + \mu)^2 - \frac{\beta^2}{2}v^2 \, .
\end{align}
If the temperature is small, ($\beta v \gg 1$), then $\sech(\beta v) \approx 2 e^{-\beta v}$, and $\tanh(\beta v) \approx 1$. Therefore,
\begin{align}
\label{eq-app:FisherQubitAux2}
(\Delta A)^2 - \tfrac{1}{2}\big\| [\sqrt{\rho},A] \big\|_2^2 &\approx 1 - \frac{1}{v^2} (\Omega_z + \mu)^2 - \frac{1}{v^2} \Omega_x^2  + \frac{2 e^{-\beta v}}{v^2}  \Omega_x^2  = \frac{2 e^{-\beta v}}{v^2}  \Omega_x^2.
\end{align}
Meanwhile, $c_2 \approx 1/2$ at high temperature ($\beta v \ll 1$), whereas $c_2 \approx e^{\beta \|H\|_s/2}/(\beta \|H\|_s)^2 = e^{\beta v}/(2\beta v)^2$ for $\beta v\gg 1$.

By Eq.~\eqref{eq-app:FisherQubitAux1}, at high temperatures ($\beta \|H\|_s \ll 1$) the 
bounds~\eqref{eq:BoundsDiagonalWY} in the main text become
\begin{subequations}
\label{eq-app:BoundsDiagonalWYhighT}
\begin{align}
\label{eq-app:upperBoundsDiagonalWYhighT}
\info  &\leq 2.4 c_2 \, \beta^2 \tr{\sqrt{\rho} \delta A \sqrt{\rho} \delta A} 
\approx \frac{2.4}{2}\beta^2 \left( 1 - \frac{\beta^2}{2} (\Omega_z + \mu)^2 - \frac{\beta^2}{2}v^2 \right) \approx 1.2 \beta^2, \\
\info  &\geq 0.8 \, \beta^2 \tr{\sqrt{\rho} \delta A \sqrt{\rho} \delta A} 
\approx 0.8 \, \beta^2 \left( 1 - \frac{\beta^2}{2} (\Omega_z + \mu)^2 - \frac{\beta^2}{2}v^2 \right) \approx 0.8 \, \beta^2.
\end{align}
\end{subequations}
By Eq.~\eqref{eq-app:FisherQubitAux2}, at low temperature ($\beta \|H\|_s \gg 1$) the 
bounds~\eqref{eq:BoundsDiagonalWY} become    
\begin{subequations}
\label{eq-app:BoundsDiagonalWYlowT}
\begin{align}
\label{eq-app:upperBoundsDiagonalWYlowT}
\info  &\leq 2.4 c_2 \, \beta^2 \tr{\sqrt{\rho} \delta A \sqrt{\rho} \delta A} 
\approx 2.4 \frac{e^{\beta v}}{4 v^2} \left( \frac{2 e^{-\beta v}}{v^2}  \Omega_x^2 \right) = 1.2\frac{\Omega_x^2}{v^4}, \\
\info  &\geq 0.8 \, \beta^2 \tr{\sqrt{\rho} \delta A \sqrt{\rho} \delta A} 
\approx 0.8 \, \beta^2 \left( \frac{2 e^{-\beta v}}{v^2}  \Omega_x^2 \right) = 1.6 \, \frac{\beta^2 e^{-\beta v} \Omega_x^2}{v^2} \, .
\end{align}
\end{subequations}

We want to compare these bounds with the values of the quantum Fisher information. Define the $\sigma_z$ eigenstates such that $\sigma_z \ket{1} = \ket{1}$ and $\sigma_z \ket{0} = -\ket{0}$.
$H = \vec{v}.\vec{\sigma}$ has the eigenvectors 
\begin{align}
\ket{+} &= \frac{1}{\sqrt{2v\big( \Omega_z + \mu +v \big)}} \Big( (\Omega_z + \mu + v) \ket{1}  +  \Omega_x   \ket{0} \Big)
\quad \text{and} \\
\ket{-} &= \frac{1}{\sqrt{2v\big( \Omega_z + \mu +v \big)}} \Big(  - \Omega_x  \ket{1}  +  (\Omega_z + \mu + v) \ket{0} \Big),
\end{align}
corresponding to eigenvalues $\pm v$. 
By the expression~\eqref{eq-app:FisherSingle} for the quantum Fisher information, for a qubit, 
\begin{align}
\label{eq-app:FisherQubitAux4}
    \info &=2\sum_{\omega_j \neq \omega_k}   \frac{ (p_j-p_k)^2}{(p_j+ p_k)(\omega_j-\omega_k)^2}  \big| \delta A_{jk}^l \big|^2   +   \sum_{\omega_j = \omega_k} \beta^2 p_j \big| \langle j | \delta A_l | k \rangle \big|^2 \nonumber \\
    &=\frac{4  }{Z} 
    \left(  \frac{ (e^{-\beta v} - e^{\beta v})^2}{(e^{-\beta v} + e^{\beta v})(2v)^2} \right)
     \big|\! \bra{+} \delta X \ket{-}\! \big|^2 +     \beta^2 \frac{e^{-\beta v}}{Z} \big| \langle + | \delta X | + \rangle \big|^2    +         \beta^2 \frac{e^{\beta v}}{Z} \big| \langle - | \delta X | -\rangle \big|^2 .
\end{align}
We evaluate this expression using
\begin{subequations}
\label{eq-app:FisherQubitAux3}
\begin{align}
\langle + | \delta A | - \rangle 
&=          \frac{1}{ 2v\big( \Omega_z + \mu +v \big) }  \Big( -\Omega_x(\Omega_z + \mu + v)  - \Omega_x (\Omega_z+\mu+v) \Big)  
= -\frac{\Omega_x}{ v } \, , \\
\langle + | \delta A | + \rangle 
&=          \frac{1}{ 2v\big( \Omega_z + \mu +v \big) } \Big(  (\Omega_z+\mu+v)^2 -\Omega_x^2  \Big) - \langle \sigma_z \rangle \coloneqq a - \langle \sigma_z \rangle,
\quad \text{and} \\
\langle - | \delta A | - \rangle &=  \frac{1}{ 2v\big( \Omega_z + \mu +v \big) } \Big(  \Omega_x^2  - (\Omega_z+\mu+v)^2 \Big)   -\langle \sigma_z \rangle \coloneqq -a - \langle \sigma_z \rangle .
\end{align}
\end{subequations}
We have defined $a \coloneqq \frac{1}{ 2v\big( \Omega_z + \mu +v \big) } \Big(  (\Omega_z+\mu+v)^2 -\Omega_x^2  \Big) = \frac{(\Omega_z+\mu)}{ v\big( \Omega_z + \mu +v \big) } \Big(   \Omega_z+\mu + v   \Big) = \frac{(\Omega_z+\mu)}{v}$.
Next, we evaluate Eq.~\eqref{eq-app:FisherQubitAux4} using Eq.~\eqref{eq-app:FisherQubitAux3}, $Z = 2 \cosh(\beta v)$, and  $\langle \sigma_z \rangle = - \frac{\tanh(\beta v)}{v} (\Omega_z+\mu)$ from Eq.~\eqref{eq-app:VarQbit}:
\begin{align}
    \info &=\frac{4 }{Z} 
    \left(  \frac{ (e^{-\beta v} - e^{\beta v})^2}{(e^{-\beta v} + e^{\beta v})(2v)^2} \right)\frac{\Omega_x^2}{v^2}
      +   \left( \beta^2 \frac{e^{-\beta v}}{Z}  +       \beta^2 \frac{e^{\beta v}}{Z} \right) (a^2+\langle \sigma_z\rangle^2) +   \left( \beta^2 \frac{e^{-\beta v}}{Z}  -       \beta^2 \frac{e^{\beta v}}{Z} \right) (-2a\langle \sigma_z \rangle) \nonumber \\
      &=  \frac{ (-e^{-\beta v} + e^{\beta v})}{2\cosh(\beta v) } \tanh(\beta v)  \frac{\Omega_x^2}{v^4}
      +   \beta^2  \frac{e^{-\beta v} + e^{\beta v}}{2\cosh(\beta v)} (a^2+\langle \sigma_z\rangle^2) -    \beta^2 \frac{e^{-\beta v}- e^{\beta v}}{2\cosh(\beta v)} 2a\langle \sigma_z \rangle \nonumber \\
      &=   \tanh^2(\beta v)  \frac{\Omega_x^2}{v^4}
      +   \beta^2 (a^2+\langle \sigma_z \rangle^2) 
      +    2 \beta^2 \tanh(\beta v)  a\langle \sigma_z \rangle  \nonumber \\
      &= \tanh^2(\beta v)  \frac{\Omega_x^2}{v^4}
      +   \beta^2  \left( \frac{(\Omega_z+\mu)^2}{v^2} + \tanh^2(\beta v) \frac{(\Omega_z+\mu)^2  }{v^2} \right)
     -  2 \beta^2 \tanh^2(\beta v)  \frac{(\Omega_z+\mu)^2}{v^2}  \nonumber \\
     &= \tanh^2(\beta v)  \frac{\Omega_x^2}{v^4}
      +   \beta^2  \frac{(\Omega_z+\mu)^2}{v^2} 
     -  \beta^2 \tanh^2(\beta v)  \frac{(\Omega_z+\mu)^2}{v^2}      \nonumber \\
     &= \tanh^2(\beta v)  \frac{\Omega_x^2}{v^4}
      +   \beta^2  \frac{(\Omega_z+\mu)^2}{v^2} \left( 1
     -   \tanh^2(\beta v)  \right).
\end{align}
Using that $\tanh(x) \approx x$ for $x<<1$ and that $\tanh(x) \approx 1$ for $x>>1$ yields
\begin{subequations}
\begin{align}
\label{eq-app:approxFInfoHighT}
\info &\approx   \frac{\beta^2 \Omega_x^2}{v^2} + \frac{\beta^2 (\Omega_z+\mu)^2}{v^2}(1-\beta^2 v^2)  
\approx \beta^2 \, , \quad \textnormal{for} \,\, \beta \|H\|_s \ll 1, \quad \textnormal{and} \\
\label{eq-app:approxFInfoLowT}
\info &\approx \frac{\Omega_x^2}{v^4}  \, ,  \quad \textnormal{for} \,\, \beta \|H\|_s \gg 1.
\end{align}
\end{subequations}

Let us compare the high-temperature upper bound~\eqref{eq-app:upperBoundsDiagonalWYhighT} with the approximate value~\eqref{eq-app:approxFInfoHighT}, as well as the low-temperature upper bound~\eqref{eq-app:upperBoundsDiagonalWYlowT} with the approximate value~\eqref{eq-app:approxFInfoLowT}. The main-text upper bound~\eqref{eq:BoundsDiagonalWY} is saturable, to within a constant multiplicative factor, in both temperature regimes.
Together with the Cram\'er-Rao bound, our bounds imply that
\begin{align}
 \textnormal{var}_{\textnormal{opt}} ( \hat\mu_l ) \approx 
 \begin{cases}& \frac{1}{\mathcal{N} \beta^2} \, ,  \quad \textnormal{for} \,\, \beta \|H\|_s \ll 1, \\
 &  \frac{\|H\|_s^4}{16 \mathcal{N} \Omega_x^2} \, , \quad \textnormal{for} \,\, \beta \|H\|_s \gg 1.
 \end{cases}
\end{align}

\section{Comparisons of bounds on the quantum Fisher information}
\label{sec-app:Comparisons}

In this section, we calculate quantum Fisher information in a spin-chain example. We compare the exact value with our bounds, Eqs.~\eqref{eq:BoundsDiagonalDeltaH} and~\eqref{eq:BoundsDiagonalWY} in the main text. We reproduce the bounds here for convenience:
\begin{subequations}
 \begin{align}
  \info_{ll}  &\leq \beta^2 \, (\Delta A_l)^2 , \\ 
  \info_{ll} &\geq 4  \beta^2 c_1  \, (\Delta A_l)^2 ,  
\end{align}   
\end{subequations}
and   
\begin{subequations}
\begin{align}
\info_{ll} &\leq 2.4 \, c_2 \, \beta^2 \, \left( (\Delta A_l)^2 \nobreak-\nobreak \tfrac{1}{2}\big\|  [\sqrt{\rho},A_l] \big\|_2^2 \right),
\quad \\
\info_{ll} &\geq 0.8 \, \beta^2 \left( (\Delta A_l)^2 \nobreak-\nobreak \tfrac{1}{2}\big\|  [\sqrt{\rho},A_l] \big\|_2^2 \right).
\end{align}
\end{subequations}
We also compare the bounds to the one derived in Ref.~\cite{miller_energytemperature}: 
\begin{align}
\info_{ll} \nobreak \leq \nobreak \beta^2 \int_0^1 \tr{\rho^a \delta A_l \rho^{1-a} \delta A_l} da,
\end{align}
with $\delta A_l \coloneqq A_l \nobreak-\nobreak \langle A_l \rangle$.

Consider estimating the parameter $\mu$ from the thermal state of a spin chain. We consider a one-dimensional chain composed of $n$ spin-$1/2$ systems, with the Hamiltonian
\begin{align}
    H = \nobreak  \mu \sum_{j=1}^n  \sigma_z^j  \nobreak + \nobreak  \lambda \sum_{j=1}^{n-1} \sigma_x^j \otimes \sigma_x^{j+1}  \eqqcolon \nobreak \mu A_\mu \nobreak+\nobreak \lambda A_\lambda.
\end{align}
$A_\mu$ and $A_\lambda$ are the operators that multiply the parameters $\mu$ and $\lambda$.

Figure~\ref{fig-app:fig3} compares the quantum Fisher information about $\mu$ with the upper and lower bounds in Eqs.~\eqref{eq:BoundsDiagonalDeltaH} and~\eqref{eq:BoundsDiagonalWY}, and with the upper bound in Ref.~\cite{miller_energytemperature}. We simulate $n=5$ spins. The figure shows that the bounds are distinct and that none of them is tighter than another in all regimes: in each subfigure, the two blue curves (upper bounds derived in this Letter) and the orange star plot (bound in Ref.~\cite{miller_energytemperature}) cross, as do the two red curves (lower bounds derived in this Letter). However, the bounds are always obeyed: the black curve (exactly calculated quantum Fisher information) always lies below the blue curves and orange star plot (upper bounds)  and above the red curves (lower bounds).

\begin{figure}[h]
 (a) Quantum Fisher information vs. $\beta$ $(\lambda/\mu = 5)$  \quad\quad \qquad\qquad(b) Quantum Fisher information vs. $\lambda$ $(\beta\mu = 0.1)$ \\
  \centering  
   \includegraphics[trim=00 00 00 00,width=0.48 \textwidth]{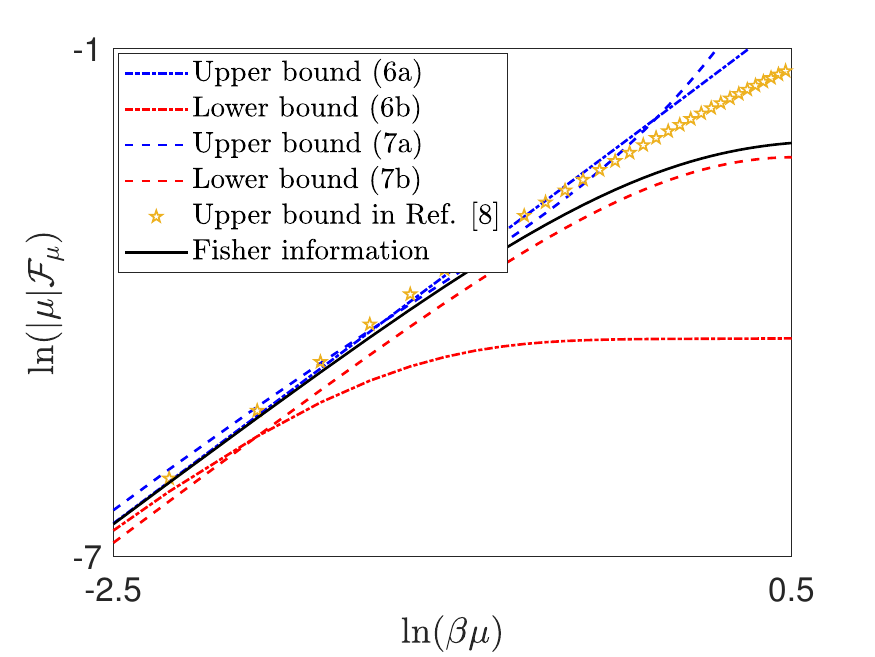} 
   \includegraphics[trim=00 00 00 00,width=0.48 \textwidth]{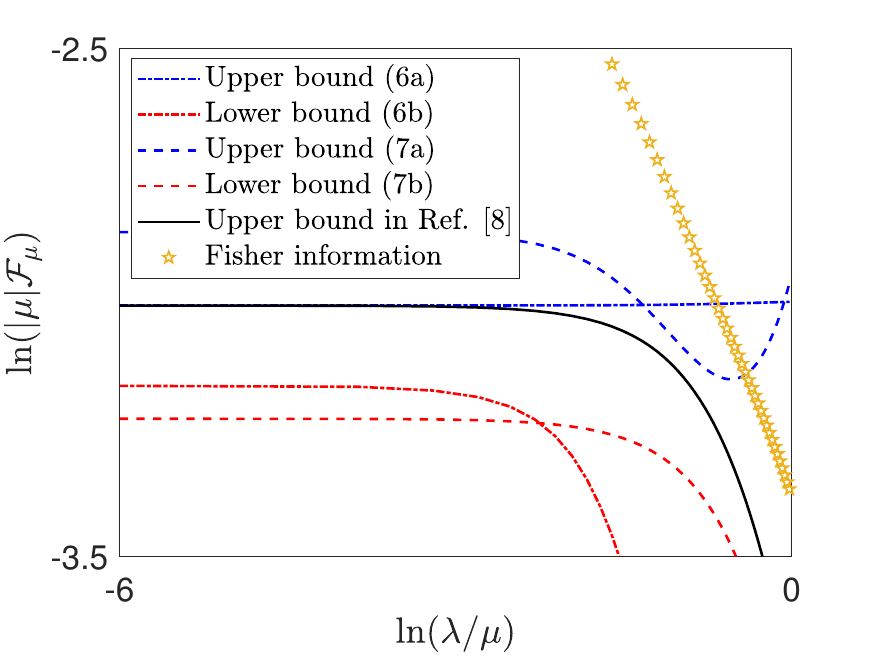} 
\caption{ 
\textbf{Comparisons of bounds on the quantum Fisher information.} 
The figure shows log-log plots of the quantum Fisher information $\info_\mu$ about parameter $\mu$, as a function of the inverse temperature (left) and as a function of the parameter $\Omega$ (right). The 5-qubit system has the Hamiltonian  $H = \nobreak  \mu \sum_{j=1}^n \sigma_z^j  \nobreak + \nobreak  \lambda  \sum_{j=1}^{n-1} \sigma_x^j \otimes \sigma_x^{j+1}  \coloneqq\nobreak \mu A_\mu \nobreak+\nobreak \lambda A_\lambda$. The plots also depict the upper and lower bounds in Eqs.~\eqref{eq:BoundsDiagonalDeltaH} and~\eqref{eq:BoundsDiagonalWY}, and the upper bound derived in Ref.~\cite{miller_energytemperature}. Each plot illustrates (where a red line crosses a red line or a blue line or orange stars cross) how different bounds can be tighter in different regimes.
\label{fig-app:fig3}
}
\end{figure}

\section{A model that can beat the standard quantum limit}
\label{sec-app:SQL} 

Here, we prove that the Hamiltonian
\begin{equation}
H = \mu \sum_{j=1}^n \left( \sigma_z^j+ 1 \right) \nobreak -\nobreak \lambda    \bigotimes_{j=1}^n n \sigma_x^j \equiv\nobreak H_\mu \nobreak+\nobreak H_\lambda,
\end{equation}
considered in the main text has the GHZ state as its unique ground state when $\mu\ll\lambda$. We also prove that $\left( \Delta H_\mu \right)^2 \nobreak\approx \nobreak\mu^2 n^2 $ for $\beta \lambda n \gg 1$.

For convenience, we shift $H_\mu$ by an irrelevant factor of $(\mu n) \id$ so that we consider the new Hamiltonian
\begin{equation}
\widetilde{H} = \mu \sum_{j=1}^n \left( \sigma_z^j+ 1 + n\right) \nobreak -\nobreak \lambda    \bigotimes_{j=1}^n n \sigma_x^j \equiv\nobreak \widetilde{H}_\mu \nobreak+\nobreak H_\lambda.
\end{equation}
In the computational basis, labeled by bit strings $s\in\{0,1\}^n$, this Hamiltonian is block diagonal, with $2^{n-1}$ blocks of dimension two spanned by pairs of computational basis states $\{\ket{s}, \ket{\overline{s}}\}$. Here, $\overline{s}$ denotes the complement of $s$---i.e. $\overline{s}_j=s_j+1$ (mod 2). These blocks, each labeled by a bitstring $s$, take the form 
\begin{align}  
\label{eq-app:SVIaux2}\widetilde{H}_s = \begin{bmatrix}
    \mu z & -\lambda n  \\
   -\lambda n & -\mu z 
\end{bmatrix},
\end{align}
where $z\equiv 2|\overline{s}|-n$ and $|\overline{s}|$ denotes the Hamming weight (i.e.~the number of ones) of the bitstring $\overline{s}$. Note that we have two distinct, but equivalent, choices of the bitstring $s$ that labels each block. Also, $z\in[-n,n]$. 

Each such block can be simply diagonalized, and, thus, so can $\widetilde{H}=\bigoplus_{s} \widetilde{H}_s$.  The  eigenvalues of Eq.~\eqref{eq-app:SVIaux2} are $\pm\sqrt{\mu^2z^2 + \lambda^2 n^2}$. Consequently, the minimum eigenvalue of $\widetilde{H}$ occurs for the block where $|\overline{s}| = 0$ (under a different, but equivalent, choice of labeling this block, $|\overline{s}| = n$). The associated minimum eigenvalue is $-n\sqrt{\mu^2 + \lambda^2}$ and the corresponding eigenstate (the ground state of $\widetilde{H}$) is 
\begin{equation}
\ket{\mathrm{gs}}\propto -\left(\frac{\mu -\sqrt{\mu^2 + \lambda^2}}{\lambda}\right)\ket{s}+\ket{\overline{s}}.
\end{equation}

Consequently, for $\mu/\lambda\ll 1$, it holds that $\ket{\mathrm{gs}}\propto \ket{s}+\ket{\overline{s}}$, which is precisely the GHZ state $\ket{\Phi}$. 

For the GHZ state, it holds that $\Delta \widetilde{H}_\mu=\mu n$. Consequently, we might expect, at least at low temperatures, that the thermal states of this model might also exhibit estimation errors that decrease faster than the standard quantum limit. This expectation can be analytically validated. In particular, a general thermal state takes the form
\begin{align}
\label{eq-app:stateAux1}
\rho = Z_\beta^{-1} e^{-\beta \widetilde{H}}&=Z_\beta^{-1} \bigoplus_s e^{-\beta \widetilde{H}_s},
\end{align}
where $Z_\beta$ is the partition function. It is easy to evaluate
\begin{align}
\label{eq-app:stateAux2}
e^{-\beta \widetilde{H}_s}=\cosh(\beta \sqrt{\mu^2z^2+\lambda^2 n^2})\id - \sinh(\beta \sqrt{\mu^2z^2+\lambda^2 n^2})\frac{(\mu z) \sigma_z-(\lambda n) \sigma_x}{\sqrt{\mu^2z^2+\lambda^2 n^2}}.
\end{align}
Consequently,
\begin{align}\label{eq-app:part-func}
Z_\beta=2\sum_s \cosh(\beta \sqrt{\mu^2z^2+\lambda^2 n^2}).
\end{align}
We can evaluate the variance $(\Delta \widetilde H_\mu)^2$ as $\sum_s ( \Delta \widetilde H_{\mu,s} )^2$ where $H_{\mu,s}=\mu z\sigma_z$ is the block of $H_\mu$ labeled by the bit string $s$. In particular,
\begin{align}
\label{eq-app:VarianceAux1}
\left(\Delta \widetilde H_{\mu,s}\right)^2 &= \mathrm{Tr}\left[\rho_s\widetilde H_{\mu,s}^2\right]-\mathrm{Tr}\left[\rho_s\widetilde H_{\mu,s}\right]^2 \nonumber \\
&=Z_\beta^{-1}\left[2\mu^2z^2 \cosh(\beta \sqrt{\mu^2z^2+\lambda^2 n^2})-\frac{4 \mu^4z^4 \sinh^2(\beta \sqrt{\mu^2z^2+\lambda^2 n^2})}{Z_\beta(\mu^2z^2+\lambda^2 n^2)}\right],
\end{align}
where we used that $\rho_s=Z_\beta^{-1}e^{-\beta \widetilde{H}_s}$ (note, $\rho=\bigoplus_s\rho_s$). 

Asymptotically in $\beta \lambda n$, we only have to consider the $z=n$ block in Eqs.~(\ref{eq-app:part-func})-(\ref{eq-app:VarianceAux1}), as $\lim_{x\rightarrow\infty} \frac{\cosh((1+\epsilon) x)}{\cosh(x)}=\infty$ (also, $\lim_{x\rightarrow\infty} \frac{\sinh((1+\epsilon) x)}{\sinh(x)}=\infty$) for any $\epsilon>0$. Consequently, asymptotically in $\beta \lambda n$,
\begin{align}
\left(\Delta \widetilde H_{\mu}\right)^2 &\sim \frac{1}{2\cosh(\beta n\sqrt{\mu^2+\lambda^2})}\left[2\mu^2n^2\cosh(\beta n\sqrt{\mu^2+\lambda^2})-\frac{2 \mu^4n^4 \sinh^2(\beta n\sqrt{\mu^2+\lambda^2})}{n^2\cosh(\beta n\sqrt{\mu^2+\lambda^2})(\mu^2+\lambda^2)}\right]\nonumber \\
&=\mu^2n^2- \tanh^2(\beta n\sqrt{\mu^2+\lambda^2})\frac{\mu^4n^2 }{\mu^2+\lambda^2} \nonumber \\
&\sim \mu^2n^2\left(1-\frac{\mu^2}{\mu^2+\lambda^2}\right),
\end{align}
where, in the last line, we use that $\beta > 0$. Therefore, we find the scaling $\left(\Delta \widetilde H_{\mu}\right)^2 \sim \mu^2 n^2$ for $\beta \lambda n \gg 1$---up to subleading, constant factor contributions to the scaling that depend on $\lambda$.
Since shifts by constants do not change the variance of an operator, this also implies that $\left(\Delta H_{\mu}\right)^2 \sim \mu^2 n^2$

\section{Saturability of the multiparameter Cram\'er-Rao bound}
\label{sec-app:saturability}

In this section, we derive conditions under which the multiparameter Cram\'er-Rao bound is saturated. That is, we prove Eq.~\eqref{eq:saturability} in the main text.
The multiparameter Cram\'er-Rao bound 
is saturable if and only if~\cite{liu2019quantum}
\begin{align}
\label{eq-app:saturability}
    \tr{\rho [L_l,L_m]}
= 0.
\end{align}
To calculate this equation's left-hand side, we express the trace relative to the $\rho$ eigenbasis.
Relative to that eigenbasis, the symmetric logarithmic derivative is represented by a matrix with elements~\cite{liu2019quantum}  
\begin{align}
\label{eq-app:SLDaux}
   \bra{j} L_l \ket{k} = 2 \frac{\bra{j} \partial_l \rho \ket{k}}{p_j + p_k} \, .
\end{align}
After substituting into the trace, we invoke Eqs.~\eqref{eq-app:statederivOffDiag} and~\eqref{eq-app:statederivDiag}:
\begin{align}
    \tr{\rho [L_l,L_m]} &= \sum_{jk} \Big( p_j \bra{j} L_l \ket{k} \bra{k} L_m \ket{j} - p_k \bra{k} L_m \ket{j} \bra{j} L_l \ket{k} \Big) = \sum_{jk} (p_j - p_k) \bra{j} L_l \ket{k} \bra{k} L_m \ket{j} \nonumber \\
    &= 4 \sum_{jk} \frac{(p_j - p_k)}{(p_j + p_k)^2}  \bra{j} \partial_l \rho \ket{k} \bra{k} \partial_m \rho \ket{j}  \nonumber \\
    &= 4 \sum_{\omega_j \neq \omega_k} \frac{(p_j - p_k)}{(p_j + p_k)^2} \bra{j} \delta A_l \ket{k} \bra{k} \delta A_m \ket{j} \frac{\left( p_j  - p_k \right)^2}{(\omega_j - \omega_k)^2} \nonumber \\
    &+ 4 \beta^2 \sum_{\omega_j = \omega_k} \frac{(p_j - p_k)}{(p_j + p_k)^2} \bra{j} \delta A_l \ket{k} \bra{k} \delta A_m \ket{j} p_j^2 \nonumber \\
    &=4 \sum_{\omega_j \neq \omega_k} \frac{(p_j - p_k)^3}{(\omega_j - \omega_k)^2 (p_j+p_k)^2} \bra{j} A_l \ket{k} \bra{k} A_m \ket{j}.
\end{align}
This expression and Eq.~\eqref{eq-app:saturability} imply Eq.~(12) 
in the main text. 

In typical Hamiltonians, most parameters will not satisfy the rather stringent conditions~\eqref{eq:saturability} for saturation. 
They are satisfied, for example, when 
the operators $A_l$ are diagonal relative to the energy eigenbasis.
Hence the multiparameter Cram\'er-Rao bound is saturable when one is estimating the Hamiltonian eigenvalues $\omega_j$.

The single parameter Cram\'er-Rao bound can be saturated with measurements in the eigenbasis of the symmetric logarithmic derivative $L_l$ in Eq.~\eqref{eq-app:SLDaux2}~\cite{paris2009quantum}.
Using Eqs.~\eqref{eq-app:statederivOffDiag} and~\eqref{eq-app:statederivDiag} into Eq.~\eqref{eq-app:SLDaux}, we find that 
  \begin{align}
  \label{eq-app:SLDaux2}
  L_l &= \sum_{\omega_j \neq \omega_k} 2 \frac{\bra{j} \partial_l \rho \ket{k}}{p_j + p_k} \ket{j}\!\bra{k}      +  \sum_{\omega_j = \omega_k} 2 \frac{\bra{j} \partial_l \rho \ket{k}}{2p_j } \ket{k}\!\bra{j} \nonumber \\
  &= 2 \sum_{\omega_j \neq \omega_k} \frac{\left( p_j  - p_k \right)}{(p_j + p_k)(\omega_j-\omega_k)} \bra{j} \delta A_l \ket{k} \ket{j}\!\bra{k}      -  \beta \sum_{\omega_j = \omega_k}  \bra{j} \delta A_l \ket{k} \ket{k}\!\bra{j}.
  \end{align}
Performing measurements on the eigenbasis of $L_l$ would yield one protocol to saturate the Cram\'er-Rao bound.

  \section{Comparisons with the Hamiltonian-learning literature}
\label{sec-app:HamiltonianLearning}

In this section, we compare our bounds to earlier results concerning the Hamiltonian-learning problem.
Two approaches to Hamiltonian learning are common: (i) the steady-state-based approach and (ii) the time-evolution-based approach. In the steady-state-based approach, one studies states $\rho$ that are stationary with respect to evolution under the Hamiltonian $H$. These steady states satisfy the condition~\cite{bairey2019learning}
$$\partial_{t}\rho=-i\left[H,\rho\right]=0 .$$ 
Every Hamiltonian eigenstate is a steady state, as is the Gibbs state, $\frac{\exp\left(-\beta H\right)}{\mathrm{Tr}\left(\exp\left(-\beta H\right)\right)}$. Several studies concern estimations of the Hamiltonian from eigenstates~\cite{garrison2018does,qi2019determining,greiter2018method,chertkov2018computational,zhu2019reconstructing,turkeshi2019entanglement,bairey2019learning,zhou2021can} or from Gibbs states~\cite{bairey2019learning,anshu2021sample,haah2022optimal,gu2022practical,sbahi2022provably}. 

In the time-evolution-based approach, one analyzes the system's time evolution under the Hamiltonian.  Several proposals concern learning the Hamiltonian from unitary dynamics~\cite{hangleiter2021precise,yu2023robust,franca2022efficient,gu2022practical,huang2023learning}. Experimental implementations~\cite{wang2017experimental,senko2014coherent} of Hamiltonian-learning protocols have been carried out, too. In the Hamiltonian-learning problem, one aims to learn the Hamiltonian $H$ from a physically relevant class of Hamiltonians, while minimizing the algorithm's run time and the number of copies of $\rho$. These two metrics are commonly known as sample complexity and time complexity, respectively.

In this work, we focus on learning about a Hamiltonian from Gibbs states. Our comparison of sample-complexity lower bounds with earlier works is presented in the context of the \emph{l}$_2$ distance error, defined via $ \epsilon = \left( \sum_{l=1}^M\left| \mu_l - \hat{\mu}_l \right|^2\right)^{\frac{1}{2}} $. Here, $\hat{\mu}_l$ denotes the estimate for $\mu_l$. The rationale for this comparison criterion is due to our adoption of the related metric $\epsilon_{\textnormal{err}}$, defined via $\sum_{l=1}^M \textnormal{var}(\hat\mu_l) \, = \,  \epsilon_{\textnormal{err}}^2$. We provide the following Lemma to relate the two error metrics. 

\begin{lemma}
For $\epsilon$ and $\epsilon_{\text{err}}$ defined as before, the
following holds.
\begin{enumerate}
\item $\text{Prob}\left(\epsilon^{2}\geq a\right)\leq\frac{\epsilon_{\text{err}}^{2}}{a}$
for all $a>0$.
\item $\text{Prob}\left(\left|\epsilon^{2}-\epsilon_{\text{err}}^{2}\right|\geq a\right)\leq\frac{\text{Var}\left(\epsilon^{2}\right)}{a}$
for all $a>0.$ 
\item For any two real numbers $b$ and $a$ such that $b\geq a$, let $\text{Prob}\left(a\leq\epsilon^{2}\leq b\right)=1$.
Then $a\leq\epsilon_{\text{err}}^{2}\leq b$.
\end{enumerate}
\end{lemma}
\begin{proof}
\textit{(Proof of part $1$)} Note that $\epsilon_{\text{err}}^{2}=\sum_{l=1}^{M}\text{var}\left(\hat{\mu_{l}}\right)$
and $\epsilon^{2}=\sum_{l=1}^{M}\left(\hat{\mu_{l}}-\mu_{l}\right)^{2}.$
Since $\hat{\mu_{l}}$ is an unbiased estimator for $\mu_{l}$, we
have $\mathbb{E}\left(\hat{\mu_{l}}\right)=\mu_{l}$ for $l\in\left\{ 1,2,\cdots,M\right\} .$

Thus, 
\begin{equation}
\text{var}\left(\hat{\mu_{l}}\right)=\mathbb{E}\left[\left(\hat{\mu_{l}}-\mathbb{E}\left(\hat{\mu_{l}}\right)\right)^{2}\right]=\mathbb{E}\left[\left(\hat{\mu_{l}}-\mu_{l}\right)^{2}\right].\label{eq:variance}
\end{equation}

Let us define a new random variable, $V_{l}=\left(\hat{\mu_{l}}-\mu_{l}\right)^{2}$.
Thus, using Eq.~(\ref{eq:variance}), we get 
\begin{equation}
\epsilon_{\text{err}}^{2}=\sum_{l=1}^{M}\mathbb{E}\left[V_{l}\right]\label{eq:exp1}
\end{equation}
and
\begin{equation}
\epsilon^{2}=\sum_{l=1}^{M}V_{l}.\label{eq:exp2}
\end{equation}

Since $V_{l}$ is a non-negative random variable, using Markov's inequality
with Eqs.~(\ref{eq:exp1}) and (\ref{eq:exp2}), we get
\[
\text{Prob}\left(\epsilon^{2}\geq a\right)\leq\frac{\epsilon_{\text{err}}^{2}}{a}
\]
for $a>0$. This completes the proof of part $1$.

\textit{(Proof of part $2$)} If $Y$ is a random variable with $\mathbb{E}\left(Y\right)=\alpha$
and $\text{Var}\left(Y\right)=\beta$, Chebyshev's inequality
says
\[
\text{Prob}\left(\left|Y-\alpha\right|\geq a\right)\leq\frac{\beta}{a} \quad \forall a>0.
\]

Applying Chebyshev's inequality to $Y=\epsilon^{2}=\sum_{l=1}^{M}V_{l}$,
we get
\begin{equation}
\text{Prob}\left(\left|\epsilon^{2}-\alpha\right|\geq a\right)\leq\frac{\text{Var}\left(\epsilon^{2}\right)}{a}\text{ \ensuremath{\quad}\ensuremath{\forall} }a>0.\label{eq:exp3}
\end{equation}

Since expectation is linear, we have 
\begin{equation}
\alpha=\mathbb{E}\left(\epsilon^{2}\right)=\epsilon_{\text{err}}^{2}.\label{eq:exp4}
\end{equation}

Using Eqs.~(\ref{eq:exp3}) and (\ref{eq:exp4}), we get the desired
result.

\textit{(Proof of part $3$)} For any random variable $Z$ and two real numbers
$a,b$ such that $b\geq a$, the following holds:
\begin{equation}
\text{Prob}\left(a\leq Z\leq b\right)=1\implies a\leq\mathbb{E}\left(Z\right)\leq b\label{eq:exp5}
\end{equation}

Substituting $Z=\epsilon^{2}$ in Eq.~(\ref{eq:exp5}) and
using Eq.~(\ref{eq:exp4}) for the expectation value of $Z$, we
get the desired result. 
\end{proof}

Distinctly from prior findings, our sample-complexity lower bound is defined by the commutativity of the Gibbs state with the terms in the Hamiltonian. 
Our approach relies on no assumptions about the Hamiltonian's structure. In contrast, earlier studies focused on low-interaction Hamiltonians: each term in the Hamiltonian is supported on a constant number of qubits. 
For a synopsis, refer to Table~\ref{table:complexity_Ham_learning}. 
\begin{table}[h]
\centering
\begin{tabular}{|c|c|c|}
\hline
Reference & Sample-complexity lower bound  & Key technique\\
\hline 
Bairey \textit{et al.}~\cite{bairey2019learning} & ? &  NA\\
\hline
Anshu \textit{et al.}~\cite{anshu2021sample} & $\Omega\left( \frac{\sqrt{M}+ \log\left(1-\delta\right)}{\beta \epsilon}\right)$ & Quantum state discrimination  \\
\hline
Sbahi \textit{et al.}~\cite{sbahi2022provably} & $?$   &  NA \\
\hline
Haah \textit{et al.}~\cite{haah2022optimal}& $\Omega\left( \frac{\exp\left(\beta\right)M}{\beta^2 \epsilon^2}\right)$  & Coding theory \\
\hline
Gu \textit{et al.}~\cite{gu2022practical} & $?$   &  NA\\
\hline
This work & $\Omega \left(\frac{ M }{ \beta^2 \, \epsilon_{\textnormal{err}}^2} \max \left\{  \min_l \frac{1}{\left(\Delta A_l\right)^2}   \, , \,  \min_l \frac{c_2^{-1}/2}{\left(\Delta A_l\right)^2 - \tfrac{1}{2} \left\| \left[ \sqrt{\rho}, A_l \right] \right\|_2^2} \right\} \right)$ &   Quantum Cram\'er-Rao bound\\
\hline
\end{tabular}
\caption{\textbf{Complexity of learning Hamiltonians via Gibbs states.} The error $\epsilon$ is the \emph{l}$_2$-distance error in the estimate of the Hamiltonian parameters. We use the related quantity $\epsilon_{\textnormal{err}}$, defined via $\sum_{l=1}^M \textnormal{var}(\hat\mu_l) \, = \,  \epsilon_{\textnormal{err}}^2$. Our sample-complexity lower bound, uniquely among the approaches, (i) is based on the commutativity of the Hamiltonian's terms with the Gibbs state and (ii) requires no assumptions about the Hamiltonian's structure. In contrast, previous studies were conducted for low-interaction Hamiltonians (each term in the Hamiltonian is supported on a constant number of qubits).  The question marks (?) indicate that no value has been reported or is available. Among the five prior studies, three provide no lower bounds on sample complexity. Therefore, the ``key technique'' is listed as NA (``not applicable'') for these studies.
}
\label{table:complexity_Ham_learning}
\end{table}

\end{document}